\documentclass[reqno,12pt,letterpaper]{amsart}
\usepackage{amsmath,amssymb,amsthm,graphicx,mathrsfs,url}
\usepackage[usenames,dvipsnames]{color}
\usepackage[colorlinks=true,linkcolor=Red,citecolor=Green]{hyperref}
\usepackage{amsxtra}
\usepackage{tikz}
\usepackage{amscd}
\usetikzlibrary{arrows.meta}
\usepackage[normalem]{ulem}
\usepackage{wasysym} 
\usepackage{graphicx}
\usepackage{subcaption}
\usepackage[normalem]{ulem}
\def\arXiv#1{\href{http://arxiv.org/abs/#1}{arXiv:#1}}
\usepackage{array}
\newcolumntype{P}[1]{>{\centering\arraybackslash}m{#1}}

\def\?[#1]{\textbf{[#1]}\marginpar{\Large{\textbf{??}}}}

\def\smallsection#1{\smallskip\noindent\textbf{#1}.}
\let\epsilon=\varepsilon 
\newcommand{\RR}{{\mathbb R}}

\newcommand{\CC}{{\mathbb C}}

\newcommand{\ZZ}{{\mathbb Z}}

\newtheorem{theo}{Theorem}
\newtheorem{prop}{Proposition}[section]

\newtheorem{lemm}[prop]{Lemma}

\numberwithin{equation}{section}

\DeclareMathOperator{\Spec}{Spec}

\let\Im=\Imag

\let\Re=\Real

\DeclareMathOperator{\tr}{tr}

\usepackage{scalerel}

\newcommand\reallywidehat[1]{\arraycolsep=0pt\relax%
\begin{array}{c}
\stretchto{
  \scaleto{
    \scalerel*[\widthof{\ensuremath{#1}}]{\kern-.5pt\bigwedge\kern-.5pt}
    {\rule[-\textheight/2]{1ex}{\textheight}} 
  }{\textheight} %
}{0.5ex}\\           
#1\\                 
\rule{-1ex}{0ex}
\end{array}
}

\begin{document}


\title[Dirac points for TBG with in-plane magnetic field]{Dirac points for twisted bilayer graphene with in-plane magnetic field}

\author{Simon Becker}
\email{simon.becker@math.ethz.ch}
\address{ETH Zurich, 
Institute for Mathematical Research, 
Rämistrasse 101, 8092 Zurich, 
Switzerland}

\author{Maciej Zworski}
\email{zworski@math.berkeley.edu}
\address{Department of Mathematics, University of California,
Berkeley, CA 94720, USA}

\begin{abstract}
  We study Dirac points of the chiral model of twisted bilayer graphene (TBG) with constant
  in-plane magnetic field. For a fixed
  small magnetic field, 
  we show that as the angle of twisting  varies between magic angles,  the Dirac points move between $ K, K' $ points
  and the $ \Gamma $ point. The Dirac points for zero magnetic field and non magic
  angles lie at $ K $ and $ K'$, while in the presence of a non-zero magnetic field 
  and near magic angles, they lie near the $ \Gamma $ point. For special directions of the magnetic field, we show that the Dirac points move, as the twisting angle varies, along straight lines and bifurcate orthogonally at distinguished points. At the bifurcation points, the linear dispersion relation of the merging Dirac points disappears and exhibit a quadratic band crossing point (QBCP). The results are illustrated
  by links to animations suggesting interesting additional structure. 
\end{abstract}
 
\maketitle 

\section{Introduction}
\label{s:int}
We consider the chiral model for twisted bilayer graphene in the form 
considered by Tarnopolsky--Kruchkov--Vishwanath \cite{magic} and then 
studied mathematically by Becker et al \cite{beta,bhz1,bhz2}. Following 
Kwan et al \cite{kps} and Qin--MacDonald \cite{wema}, see also \cite{RY} for bilayer graphene, we introduce an additional 
term $ B = b e^{ 2 \pi i \theta } $ corresponding to an in-plane magnetic field of strength $ b $
and direction $ 2 \pi \theta $ -- see \eqref{eq:defDB}. 

We concentrate on the case of {\em simple} magic $ \alpha $'s ($ \alpha $ is a dimensionless 
parameter roughly corresponding to the reciprocal of the angle of twisting of the 
two graphene sheets; see \S \ref{s:mult} for the discussion of simplicity). For the 
Bistritzer--MacDonald potential $ U_{\rm{BM}} ( z ) $ (see the caption to Figure \ref{f:1})
the real magic angles are expected to be simple (see Remark 1 after Theorem \ref{th:Gamma}).

\begin{figure}
\includegraphics[width=7.6cm]{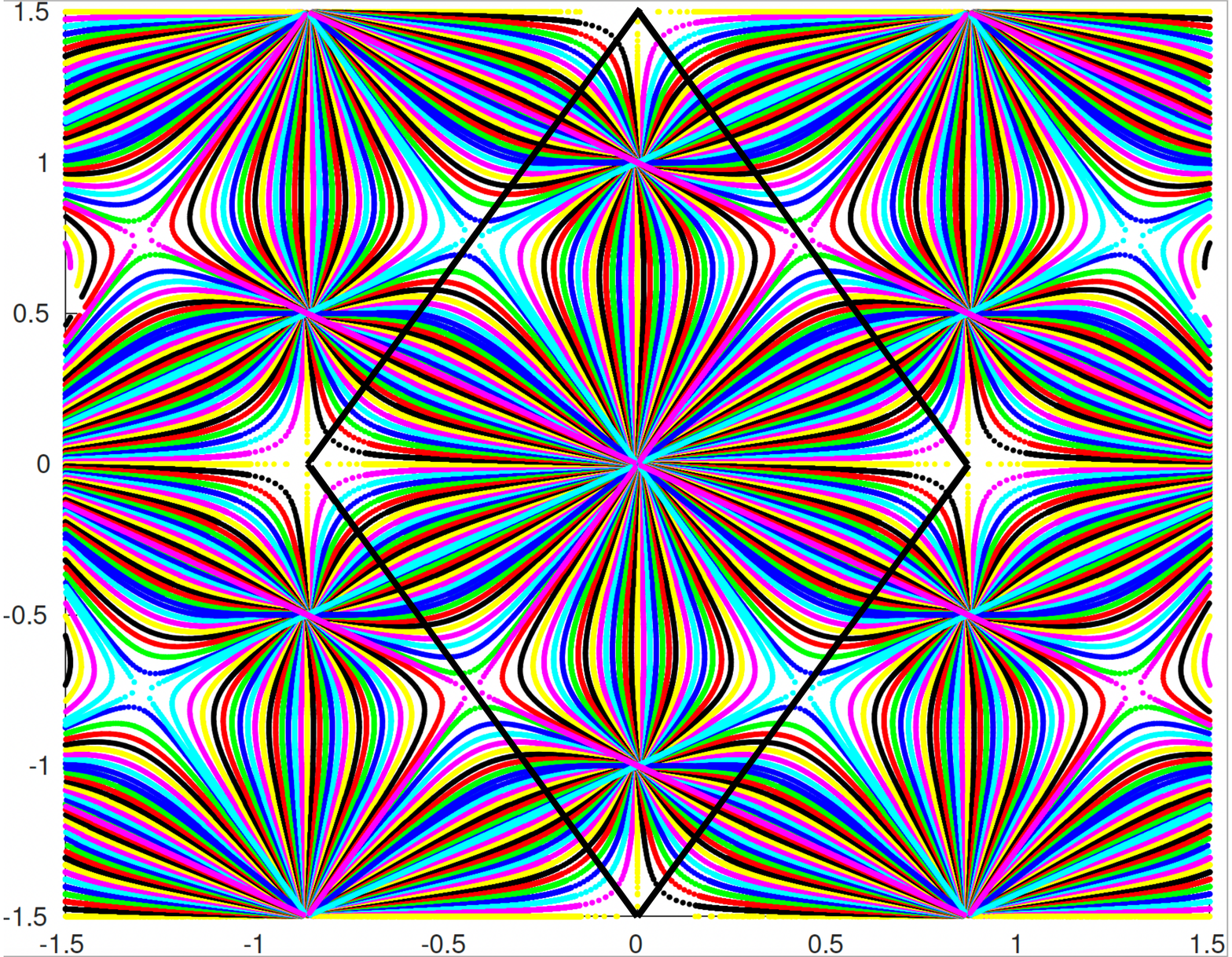}
\includegraphics[width=7.6cm]{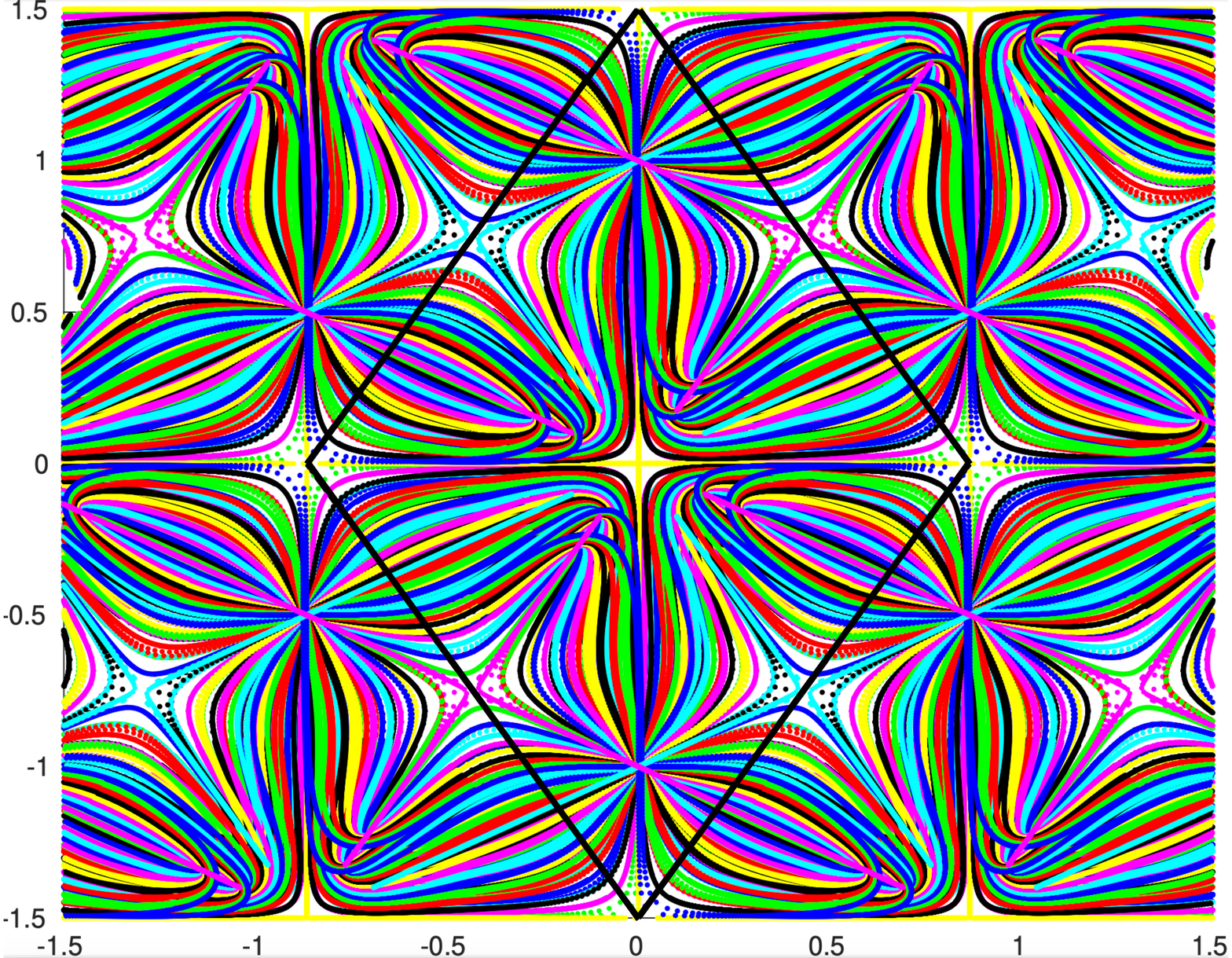}
\caption{\label{f:1} We 
show the movement  of Dirac points as  $\alpha $ varies in $ (0,2.3)$ 
for the Bistritzer--MacDonald potential $ U ( z ) = U_{\text{BM}} = \sum_{k=0}^2 \omega^k e^{\frac12 ( z\bar \omega^k - \bar z \omega^k )} $
 (left) and $\alpha \in (0,2.7)$
 and $ U ( z ) = 2^{-\frac12} ( U_{\text{BM}} (z) - \sum_{k=0}^2 \omega^k e^{-z\bar \omega^k - \bar z \omega^k}) $
  (right). (Here we use the convention of \cite{magic,beta} -- see \eqref{eq:z2zeta}.) The magnetic field is given by 
$B=B_0 e^{2\pi i \theta}$ with $ B_0 = 0.1 $ and curves of different colour correspond to 
different  $\theta \in [0,\tfrac12]$.
In the case on the left $ \alpha $ passes two simple magic $\alpha$'s; on the right, it passes two double magic $ \alpha$'s. The $ \Gamma $ point corresponds to $ 0 $ and $ K $, $ K' $ points to 
$ \pm i $. The boundary of the Brillouin zone, a fundamental domain of $ \Lambda^* $, is outlined in black. See 
\url{https://math.berkeley.edu/~zworski/B01.mp4} and
\url{https://math.berkeley.edu/~zworski/B01_double.mp4}
for the corresponding animations.
}
\end{figure}

We have the following combination of mathematical and numerical observations:

\begin{itemize}

\item We show (Theorem \ref{th:Gamma}) that a small in-plane magnetic field
destroys flat bands corresponding to simple magic $ \alpha$'s (under an 
additional non-degeneracy assumption);

\medskip

\item For small magnetic fields, the motion of Dirac points appears quasi-periodic for
$ \alpha \in [ \alpha_j , \alpha_{j+1} ] $ where $ \alpha_j $ are the magic angles for
the Bistritzer--MacDonald potential \cite{magic}. That is most striking for 
$ \theta = 0 , \frac23 $ for which the motion is linear -- see the Remark after 
Theorem \ref{t:symD} and also Figure \ref{fig:Bifurcation}.

\medskip

\item Theorem \ref{t:symD} shows that most of the action 
takes place near the magic angles: the Dirac points
get close to $ \Gamma $ point (Theorem \ref{th:Gamma}; they meet there for $ \theta = 0 $, Proposition \ref{p:SB} and $\theta = \frac23 $, Proposition \ref{p:SB1}) at simple magic angles
-- see \url{https://math.berkeley.edu/~zworski/magic_billiard.mp4} for an animation. When the Dirac cones meet, they exhibit a quadratic band crossing point (QBCP), see Figure \ref{f:c2w} and Proposition \ref{prop:well} 
(its formulation requires introduction of Bloch--Floquet spectra in \S \ref{s:BF}) -- 
for the discussion of such phenomena in the physics literature see \cite{dGGM12,kps,MLFP}.

\medskip

\item Figure \ref{f:2} (right) shows that for fixed $ \alpha$'s and varying directions of the
magnetic field, we have ``fixed points" at $ \Gamma $ and $ K , K' $ with ``normal 
crossings" and the vertices and middle of points of edges of the boundary 
of the Brillouin zone. These points are precisely the intersection of the rectangles
(other than $ \Gamma, K , K' $). 

\item The situation is more complicated for double (protected) magic angles: 
see the right panel in Figure \ref{f:1}: at magic $\alpha$'s, Dirac points are now
close to $ K$ and $K'$.

\end{itemize}

\begin{figure}
{\begin{tikzpicture}
\node at (-2,0) {\includegraphics[width=7.6cm]{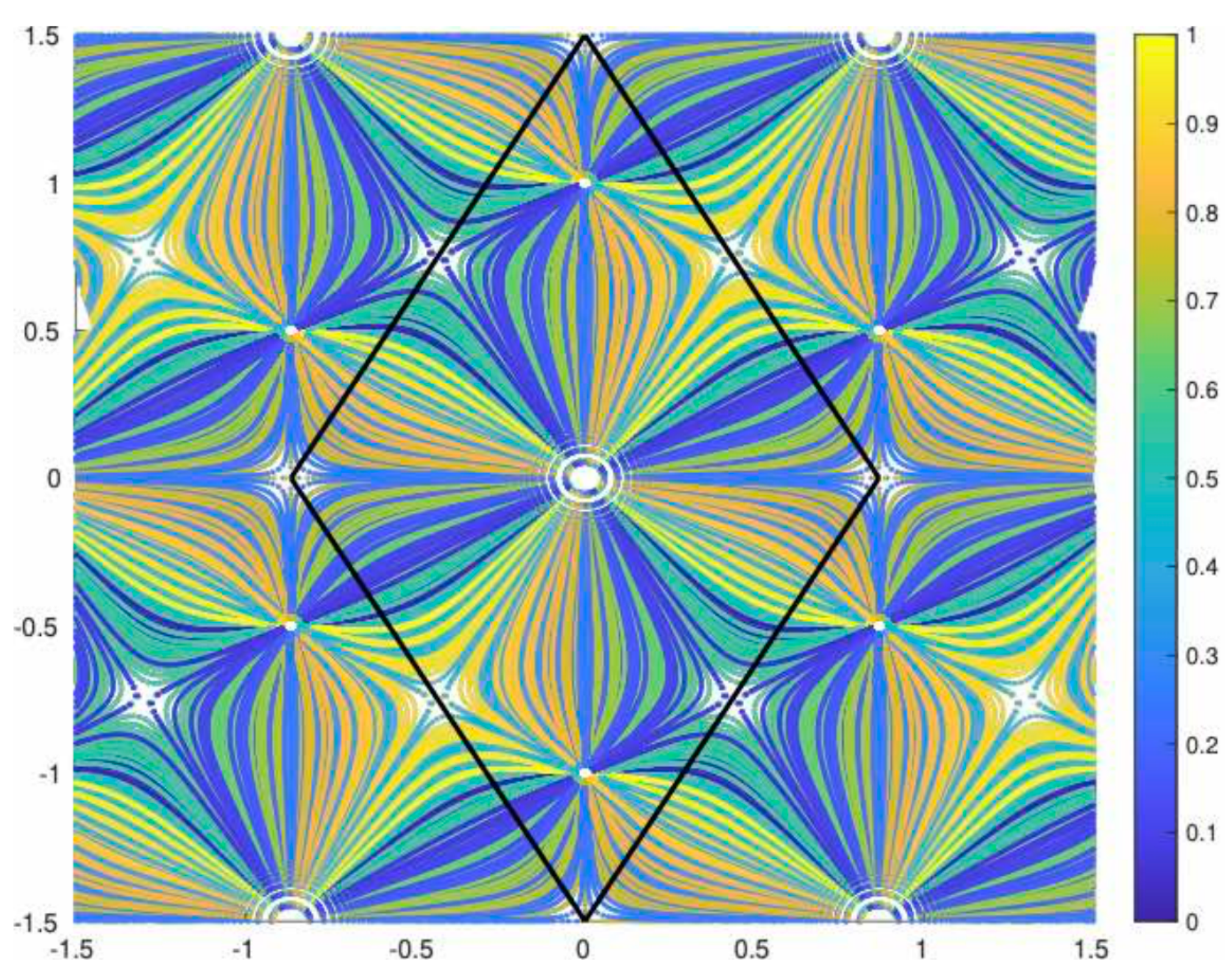}};
\node at (7.6-2,0) {\includegraphics[width=7.6cm]{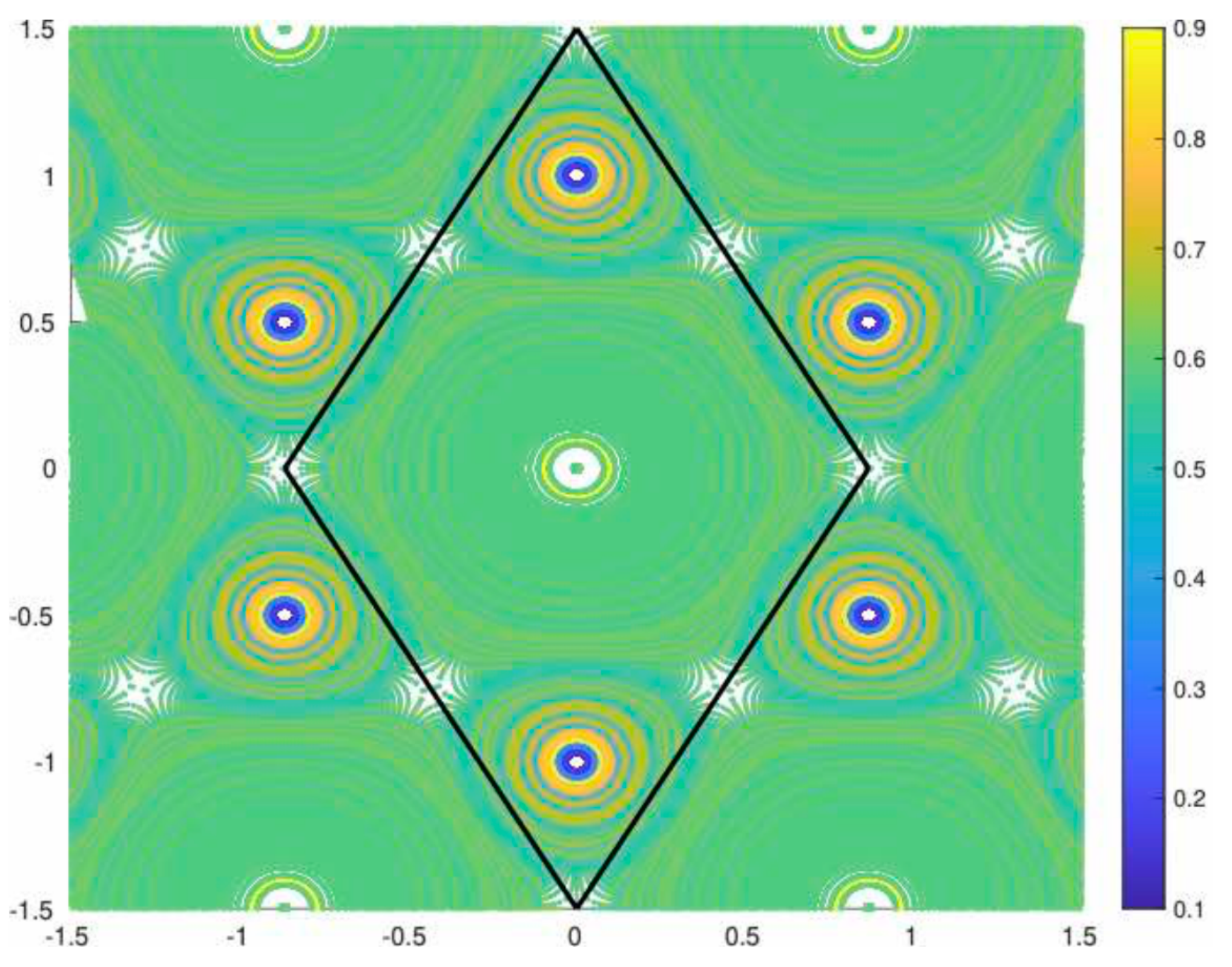}};
\node at (1.3,-3) {$\theta$}; 
\node at (7.6+3.5-2.25,-3.0) {$\alpha$};
\end{tikzpicture}}
\caption{\label{f:2} The dynamics of Dirac points 
for the Bistritzer--MacDonald potential $ U ( z ) = U_{\text{BM}} = \sum_{k=0}^2 \omega^k e^{\frac12 ( z\bar \omega^k - \bar z \omega^k )} $. The magnetic field  given by 
$B=B_0 e^{2\pi i \theta}$ with $ B_0 = 0.1 $ 
On the left different colours correspond to different values of $\theta $ shown in the
colour bar and $ \alpha $ varies between $ 0.1 $ and $ 0.9 $ (this is a colour map version of
the left panel of Figure \ref{f:1}). On the right, the colours correspond to different 
values of $ \alpha $ shown in the colour bar and $ \theta $ varies. The predominance of green 
(corresponding to the range between 0.5 and 0.6) means that most of the motion happens
near the (first) magic alpha -- see \url{https://math.berkeley.edu/~zworski/first_band.mp4} 
for $ E_1 ( \alpha, k ) /\max_k E_1( \alpha, k )$ for fixed $ B $ as $ \alpha $ varies. 
}
\end{figure}

This note is organized as follows: we present the Hamiltonian and the definition of 
Dirac points in \S \ref{s:inplane}. We also establish basic symmetry properties
 of Dirac points and a perturbation result valid away from magic $ \alpha$'s. The next section
 reviews the theory of magic angles following \cite{beta,bhz2} but in a more invariant
 and general way. In \S \ref{s:proof} we set up Grushin problems needed for the
 understanding the small in-plane magnetic fields as a perturbation. 
 We then specialize, in \S \ref{s:bif}, to directions of the magnetic field for which the Dirac
 points move linearly  as $ \alpha $ changes. In particular, they meet at 
 special points and we describe the resulting quadratic band crossing.
 We conclude in \S \ref{s:proofs} with the proofs of the main theorems.

\smallsection{Acknowledgements} 
We would like to thank Allan MacDonald for suggesting the in-plane magnetic field
problem to us, and Charles Epstein for a helpful discussion. 
MZ was partially supported by National Science
Foundation under the grant DMS-1901462 and by the Simons Foundation Targeted Grant Award No.
896630.

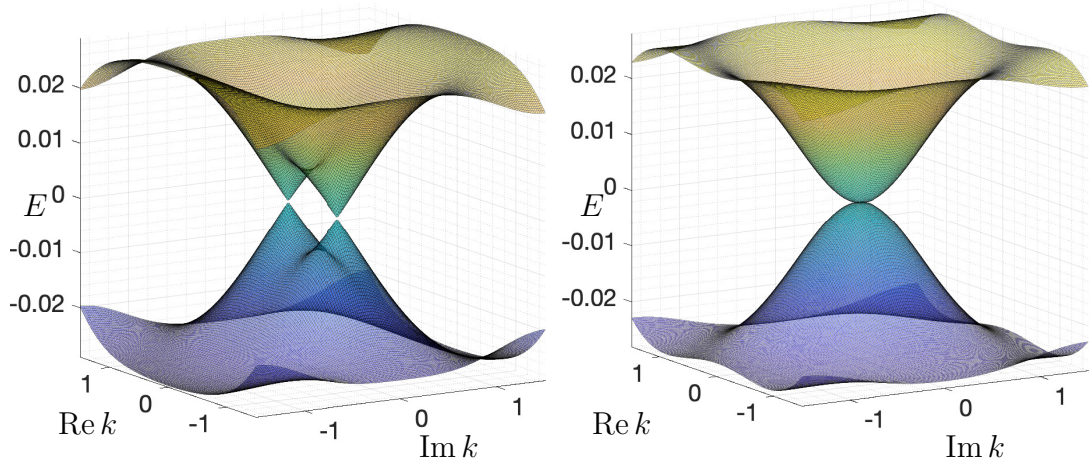
\begin{figure}
{\begin{tikzpicture}
   \node at (-11,0) {\includegraphics[trim={2.7cm 0 0 0},width=6.2cm]{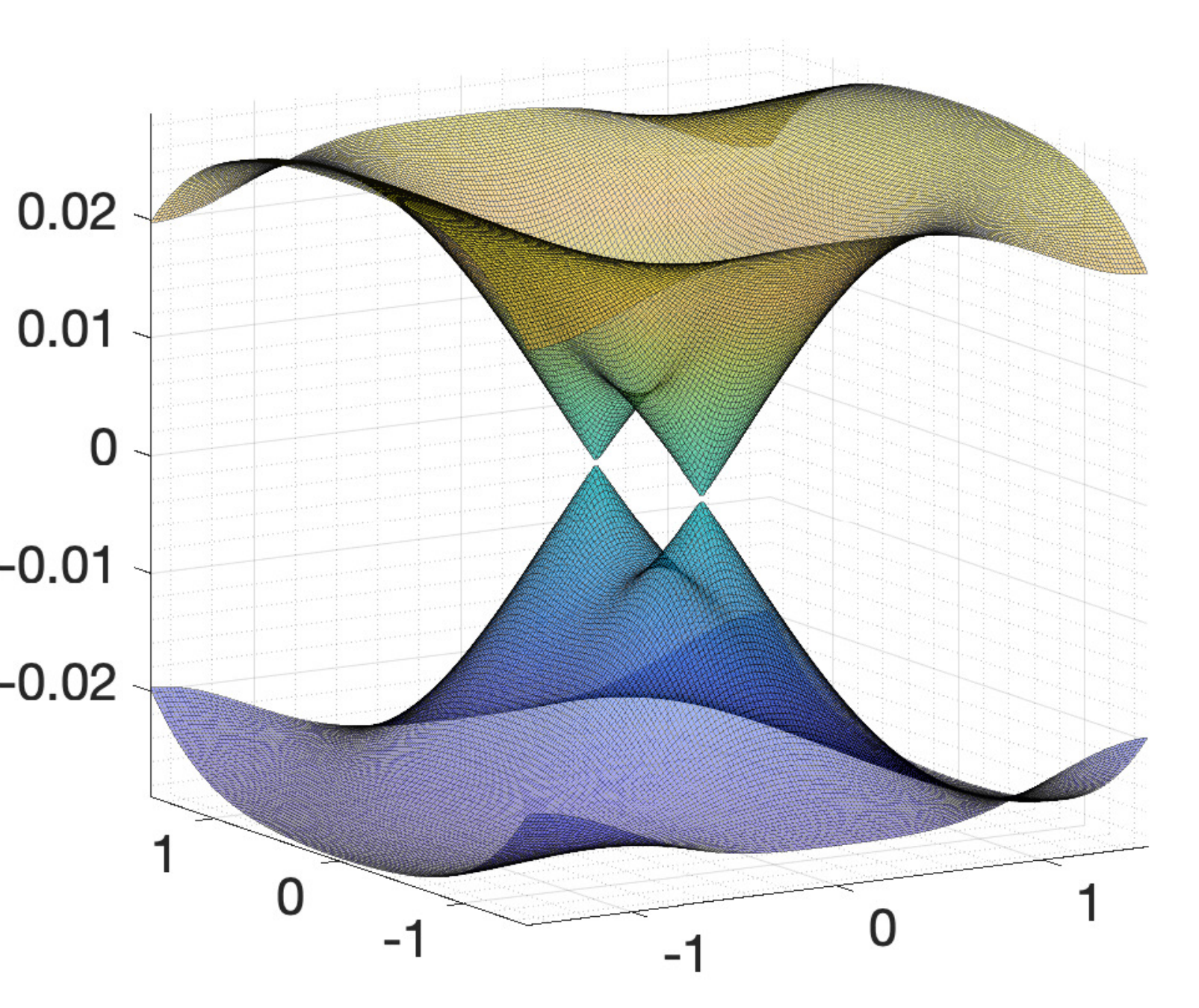}};  
    \node at (-4,0){\includegraphics[trim={1.7cm 0 0 0},width=6.5cm]{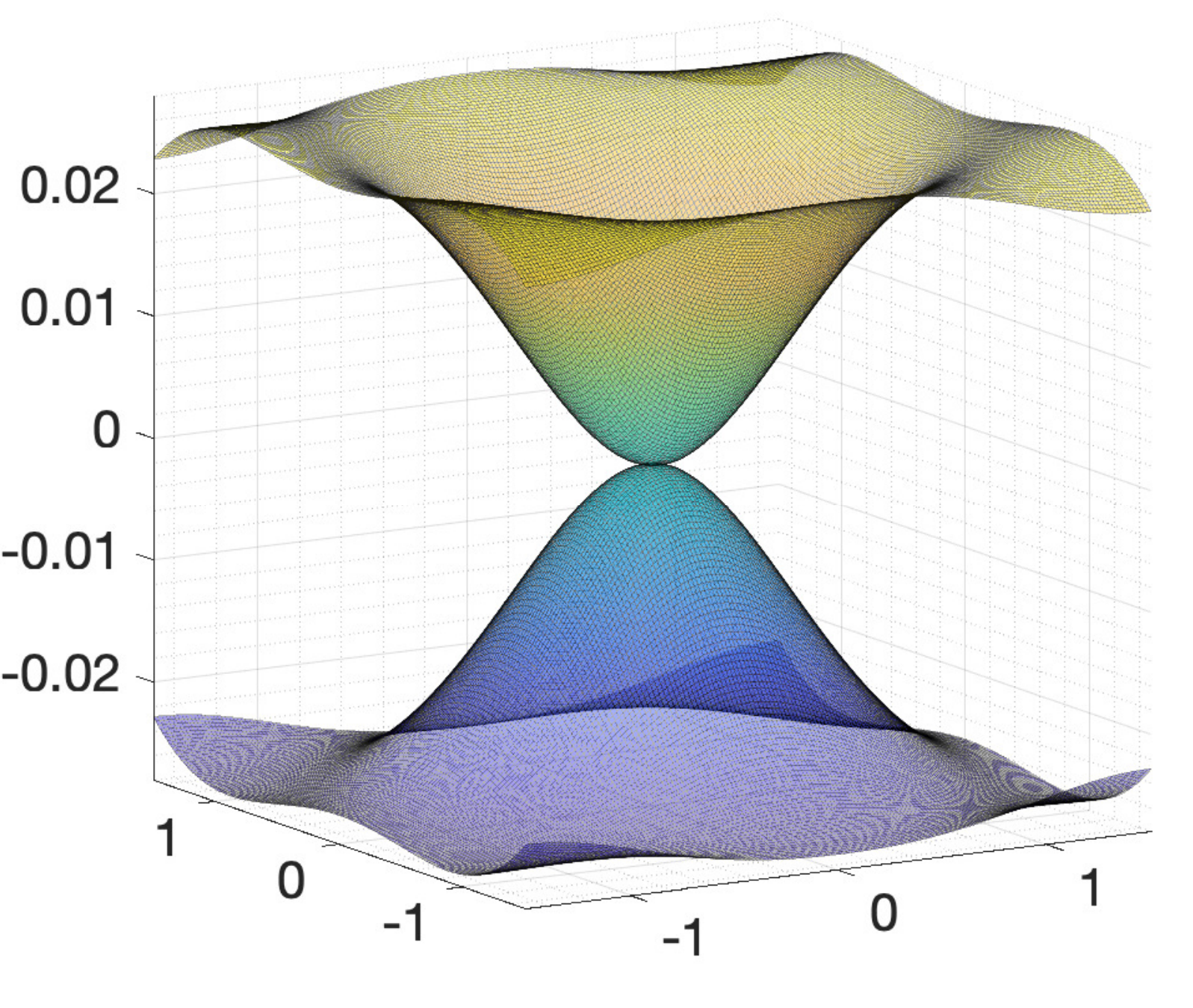}};
     \node at (-3.5+1,-3){$\Im k$};
     \node at (-8+0.8,-2.7){$\Re k$};
     \node at (-9+1.4,0.2){$E$};
       \node at (-3.5-6,-3){$\Im k$};
     \node at (-8-6.3,-2.7){$\Re k$};
     \node at (-16+1,0.2){$E$};
  \end{tikzpicture}}
\caption{\label{f:c2w} 
When $ B $ is real
(in the convention of \eqref{eq:defDB}), two Dirac cones approach $ \Gamma $ point as $ \alpha \to \alpha^* = \underline \alpha + \mathcal O ( B^3 ) $ ($ \underline \alpha $ a simple
real magic parameter) on the line $ \Im k= 0 $ (left). For 
$ \alpha = \alpha^* $, the quasi-momentum $ k $ at which the bifurcation happens are 
the boundary of the Brillouin zone and the $ \Gamma $-point which is shown in the figure (right).
The animation \url{https://math.berkeley.edu/~zworski/Rectangle_1.mp4} shows the motion of 
Dirac points in this case.}
\end{figure}

\section{In-plane magnetic field}
\label{s:inplane}
Adding a constant in-plane magnetic field \cite{kps,wema} with magnetic vector potential $A = z_{\perp}B\times \hat e_{z_{\perp}}$, where $z_{\perp}$ is the coordinate perpendicular to the two-dimensional plane of TBG and $\hat e_{z_{\perp}}$ the unit vector pointing in that direction, to the chiral model of TBG \cite{magic} results for layers at positions $z_{\perp} = \pm 1$,
in the Hamiltonian $H_B(\alpha) $ in \eqref{eq:defHB} build from non-normal operators
\begin{equation}
\label{eq:defDB} D_B ( \alpha ) := D( \alpha ) + \mathcal B , \ \ \ D ( \alpha ) = \begin{pmatrix}
2 D_{\bar z } & \alpha U ( z ) \\
\alpha U ( - z )& 2 D_{\bar z } \end{pmatrix} , \ \ \ 
\mathcal B := \begin{pmatrix} B & \ 0 \\
0 & - B \end{pmatrix} , \end{equation}
where we make the following assumptions on $ U $:
\begin{equation}
\label{eq:defU}
\begin{gathered} 
  U ( z + \gamma ) = e^{ -2  i \langle \gamma , K  \rangle } U ( z ) , \ \ U ( \omega z ) = \omega U(z) , \ \ \overline {U ( \bar z ) } = - U ( - z ) , \ \ \ 
\omega = e^{ 2 \pi i/3},  \ \\ \gamma \in \Lambda := \omega \mathbb Z \oplus \mathbb Z , \ \ 
\omega K \equiv K \not \equiv 0 \!\!\! \mod \Lambda^* , \ \ 
\Lambda^* :=  \frac {4 \pi i}  {\sqrt 3}  \Lambda , \ \ \langle z , w \rangle := \Re ( z \bar w ) .
\end{gathered}
\end{equation}
In this convention the Bistritzer--MacDonald potential used in \cite{magic,beta} corresponds to
\begin{equation}
\label{eq:BMU}
U ( z ) =  - \tfrac{4} 3 \pi i \sum_{ \ell = 0 }^2 \omega^\ell e^{ i \langle z , \omega^\ell K \rangle }, \ \ \ K = \tfrac43 \pi  . 
\end{equation}

For a discussion of a perpendicular constant magnetic field in the chiral model of twisted bilayer graphene we refer to \cite{BKZ}.

\noindent
{\bf Remark.}
We adapt here a more mathematically straightforward convention  of coordinates than that of  \cite{beta,bhz1}
where we followed \cite{magic} (with some, possibly also misguided, small changes; our motivation 
comes from a cleaner agreement with theta function conventions). 
The translation between the two conventions is as follows: the operator considered in \cite{beta}, and rigorously derived in \cite{CGG, Wa22} was
\begin{equation}
\label{eq:oldD}  \begin{gathered}  \widetilde D ( \alpha ) := 
\begin{pmatrix}
2 D_{\bar \zeta } & \alpha U_0 ( \zeta ) \\
\alpha U_0 ( - \zeta )& 2 D_{\bar \zeta } \end{pmatrix} , \ \ \ \overline{ U_0 ( \bar \zeta ) } = U_0 ( \zeta ) , \\
 U_0 \left( \zeta + \tfrac{ 4 \pi i } 3 ( a_1 \omega + a_2 \omega^2  ) \right) = 
\bar \omega^{a_1 + a_2 } U_0 ( \zeta ) , \ \ \  U_0 ( \omega \zeta ) = \omega U_0 ( \zeta). \end{gathered} 
\end{equation}
We then have a (twisted) periodicity with respect to $ \tfrac13 \Gamma$ and periodicity with respect to 
\[   \Gamma := 4 \pi i ( \omega \mathbb Z + \omega^2 \mathbb Z ) = 4 \pi i \Lambda \text{ such that } \Gamma^*:=\frac{1}{\sqrt{3}} (\omega \ZZ \oplus \omega^2 \ZZ) = \frac{\Lambda}{\sqrt{3}} \]
This means that to switch to (twisted) periodicity with respect to $ \Lambda $ we need a
change of variables:
\begin{equation}
\label{eq:z2zeta}     \zeta = \tfrac 4 3 \pi i z , \ \   \tfrac13  \Gamma = \tfrac 4 3 \pi i \Lambda, \ \ 
3 \Gamma^* = ( \tfrac13 \Gamma)^* = \sqrt 3 \Lambda = \frac{ 3 }{ 4 \pi i }  \Lambda^* 
. \end{equation}
Then 
\begin{equation}
\label{eq:unit} \begin{gathered}  
\widetilde D ( \alpha ) = - \frac{ 3}{ 4 \pi i } 
\begin{pmatrix}
2 D_{\bar z  }  & \alpha  U ( z )  \\
\alpha U ( - z ) & 2 D_{\bar z }  \ \end{pmatrix}    , 
\ \ \ U ( z ) := - \tfrac43  \pi i 
U_0 \left(  \tfrac 43 \pi i z  \right).
\end{gathered} \end{equation}
The twisted periodicity condition in \eqref{eq:oldD} corresponds to the condition 
in \eqref{eq:defU} since 
 $ \bar \omega^{a_1 + a_2} = 
e^{ i \langle a_1 \omega + a_2 \omega^2 , K \rangle } $, $ K
= 4 \pi i (- \frac13 - \frac23 \omega )/\sqrt 3 = 4 \pi/3 $. 
See the caption to Figure \ref{f:1} for examples of $ U_0 ( z ) $ in the coordinates of \cite{magic,beta}.

\medskip

The self-adjoint Hamiltonian built from \eqref{eq:oldD} is given by 
\begin{equation} \label{eq:defHB} 
H_B ( \alpha ) = \begin{pmatrix} 0 & D_B ( \alpha )^* \\ D_B ( \alpha) & 0 \end{pmatrix}
\end{equation}
and the Dirac points are given by the spectrum of 
\[  D_B ( \alpha ) : H_0^1 \to L^2_0 , \ \ \ L^2_0 := \{ u \in L^2_{\rm{loc}} ( \mathbb C ; \mathbb C ) :
u ( x + \gamma ) = {\rm{diag}} ( e^{ - i \langle \gamma , K \rangle }  , 
e^{  i \langle \gamma, K \rangle } ) u ( x ) \}, \]
with a similar definition of $ H_0^1 $ (replace $ L^2_{\rm{loc}}$ with $ H^1_{\rm{loc}}$) -- see
\S \ref{s:BF} for a systematic discussion and explanations.

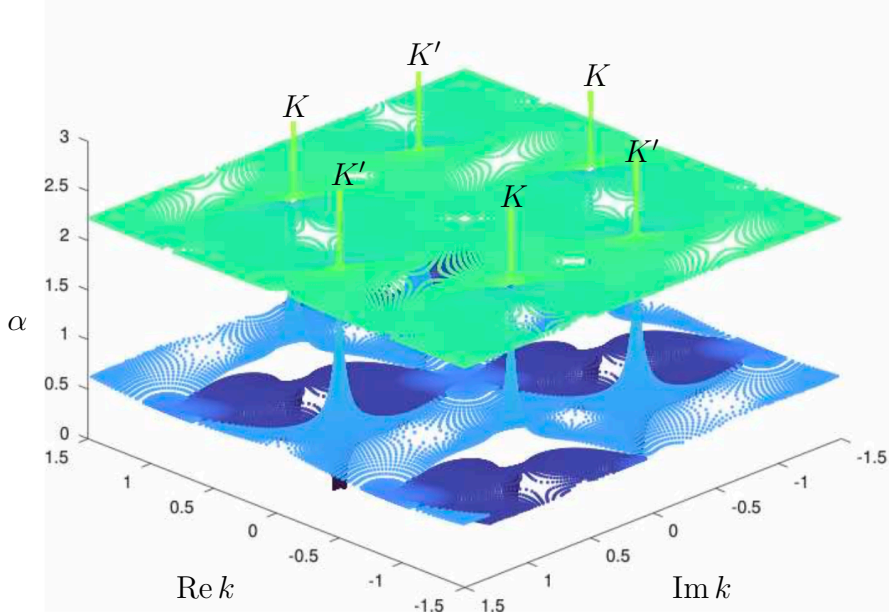
\begin{figure}
{\begin{tikzpicture}
   \node at (0,0) {\includegraphics[width=12cm]{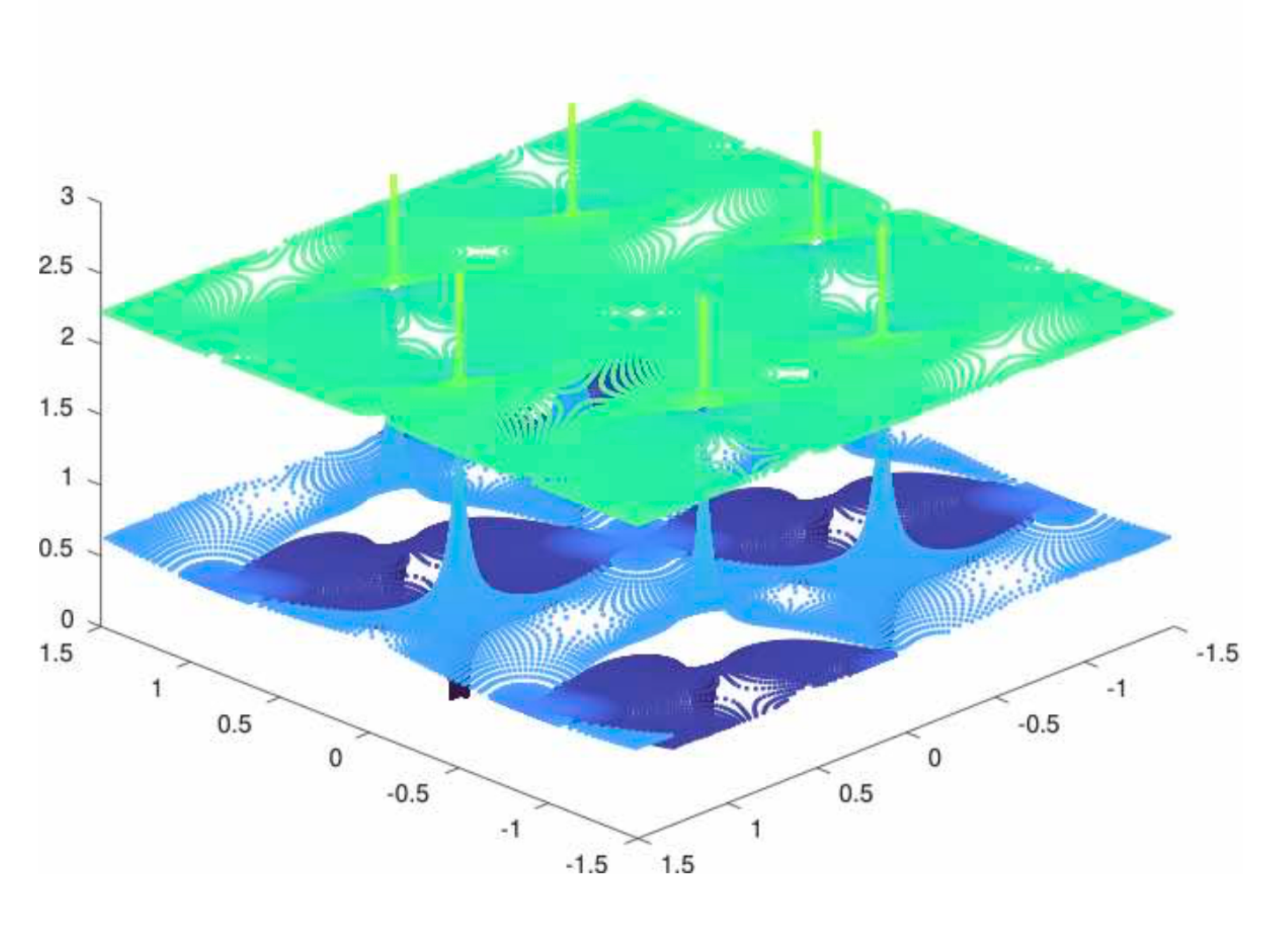}};  
     \node at (-0.6,3.6){$K'$};
     \node at (2.3,2.3){$K'$};
     \node at (-1.6,2){$K'$};
      \node at (1.7,3.3){$K$};
       \node at (0.6,1.7){$K$};
        \node at (-2.3,2.9){$K$};
    \node at (3.1,-3.5){$\Im k$};
   \node at (-3.5,-3.5){$\Re k$};
    \node at (-6,0){$\alpha$};
  \end{tikzpicture}}
  \caption{Dirac point dynamics for $B=0.1 e^{2\pi i \theta}$ with $\theta \in [0,1/2]$. Close to the first two magic angles ($\alpha \approx 0.585,2.221$), the dynamics spreads out in space.}
\end{figure}

We recall (see \S \ref{s:magic}) that there exists a discrete set $ \mathcal A \subset \mathbb C $ 
such that 
\begin{equation}
\label{eq:magic_angles}
    \Spec_{L^2_0 } ( D_0 ( \alpha ) ) = \left\{ \begin{array}{ll} ( K +  \Lambda^*) 
\cup ( - K + \Lambda^* ) & \alpha \notin \mathcal A  \\ 
\ \ \ \ \ \ \ \ \ \ \mathbb C & \alpha \in \mathcal A .  \end{array} \right. 
\end{equation}
The elements of $ \mathcal A $ are reciprocals of {\em magic angles} and the real ones
 are of physical interest. As recalled in Proposition \ref{p:magicD}, elements of 
 $ \mathcal A $ are characterized by the condition that $ \alpha^{-1} \in \Spec_{L^2_0}
 T_k $, where $ \mathbb C \setminus \{ K , - K \} \mapsto  T_k $
 is a (holomorphic) family of compact operators given in \eqref{eq:defTk} (the spectrum is independent
 of $ k $ and so are its algebraic multiplicities). In this paper we will use the following notion of simplicity (see also \S \ref{s:mult}):
 \begin{equation}
 \label{eq:strongm}
 \text{ $ \alpha \in \mathcal A $ is said to be simple } \ \Longleftrightarrow \ 
 \text{ $ 1/\alpha $ is a simple eigenvalue of $ T_k $.} \end{equation}
Here simplicity of an eigenvalue is meant in the algebraic sense.

The first result is a consequence of simple perturbation theory and of 
symmetries of $ D_B ( \alpha )$: 
\begin{theo}
\label{t:symD}
Suppose that $ \Omega \Subset \mathbb C  \setminus \mathcal A $ is an open set. 
Then there exists $ \delta = \delta ( \Omega ) $ such that for $ |B | < \delta  $ there exists
$ \alpha \mapsto k_B ( \alpha)  \in C^\omega ( \Omega  ) $ such that
\[   \Spec_{ L^2_0 } ( D_B ( \alpha ) )  = ( k_B ( \alpha  ) + \Lambda^*  ) \cup 
( - k_B ( \alpha) + \Lambda^* ) , \] 
and $ k_B  ( \alpha  ) = K + \mathcal O ( B ) $. In addition, for 
$ \alpha, B \in \mathbb C $, 
\begin{equation}
\label{eq:symD}
\begin{split} 
&    \Spec_{ L^2_0 } D_{\omega B}  ( \alpha )  = \omega \Spec_{L^2_0 } D_B ( \alpha )  , \\
&  \Spec_{ L^2_0 } D_{ B}  ( - \alpha )  =   \Spec_{L^2_0 } D_B ( \alpha )  = - \Spec_{L^2_0 } D_B ( \alpha ), \\
& \Spec_{L^2_0 } D_{  \bar B } ( \bar \alpha ) =  \overline{ \Spec_{ L^2_0 } D_{   B}  ( \alpha )} . 
\end{split} 
\end{equation}
\end{theo}

\begin{proof}[Proof of Theorem \ref{t:symD}]
Proposition \ref{p:magicD} shows that for $ \alpha \in \Omega $ the spectrum of
$ D ( \alpha ) $ is given by $ \pm K + \Lambda^* $ and for small
$ B $ we have two eigenvalues for $ D_B ( \alpha ) $.
The structure of $ D ( \alpha ) $ implies that 
\begin{equation}
\label{eq:defEsc}  \mathscr E D ( \alpha ) = - D ( \alpha ) \mathscr E , \ \ \
\mathscr E v ( z) := J v ( - z ) ,  \ \  J := \begin{pmatrix}   0 & -1 \\ 1 & \ 0 \end{pmatrix}\end{equation}
and since $ J \mathcal B = - \mathcal B J $ we also have 
\begin{equation}
\label{eq:reflection}  \mathscr E  (D_B ( \alpha )+k) \mathscr E^* = - (D_B ( \alpha )-k), \end{equation}
that is the spectrum is invariant under reflection $ k \mapsto - k $. 

Since $ \mathcal R D ( \alpha ) \mathcal R^* = \omega D ( \alpha ) $, $\mathcal R u ( z ) := u ( \omega z ) $, 
 we have 
$ \mathcal R D_B ( \alpha )\mathcal R^* = \omega D_{ \bar \omega B } ( \alpha ) $
which gives
the first identity in \eqref{eq:symD}.
 We now recall the following antilinear symmetries:
\begin{equation}
\label{eq:antil} 
\begin{gathered} F D  ( \alpha ) F =  D ( - \bar \alpha ) , \ \ \
  F {v} ( z ) := \overline{ {v} ( - \bar z )} , \\
 \mathscr Q D ( \alpha ) \mathscr Q  = D ( - \alpha )^* , \ \ \  Q v ( z ) : = \overline{ v ( - z ) } , \ \ \ 
 \mathscr Q := \begin{pmatrix} Q & 0 \\
 0 & - Q \end{pmatrix} . 
 \end{gathered} 
 \end{equation}
Since $ \bullet  \mathcal B = \mathcal B^*  \bullet $, $ \bullet = F, \mathscr Q $, 
we have 
\[ F (  D_{{\bar B}  }   ( - \bar \alpha ) - \bar k ) F=  ( D_B ( \alpha ) - k )   = 
 \mathscr Q (  D_{ B } ( - \alpha )^* - \bar k ) \mathscr Q ,  \ \  \ \mathscr Q^2 = F^2 = I , \]
 which shows that (since the spectrum is invariant under $ k \mapsto - k $), 
\begin{equation}
\label{eq:specs}   \Spec_{L^2_0}  ( D_B ( \alpha ) ) = \overline{ \Spec_{L^2_0} ( D _{ \bar B } ( -\bar \alpha ) )} = 
 \Spec_{L^2_0} (D_{  B } ( - \alpha )) ,  \end{equation}
and that gives the rest of \eqref{eq:symD}.
\end{proof} 

We now state a result valid near simple $ \underline \alpha \in \mathcal A $. 
\begin{theo}
\label{th:Gamma}
Suppose $ \underline \alpha \in \mathcal A $ is simple and $ g_0 ( \underline{\alpha} ) \neq 0 $
where $ g_0 $ is defined in \eqref{eq:G2g}. 
Then there
exists $ \delta_0 > 0 $  such that for $ 0 < |B | < \delta_0 $ and  $  | \alpha - \underline{\alpha } | <  \delta_0  $,
the spectrum of $ D_B ( \alpha )$ on $ L^2_0 $ is discrete and 
\begin{equation}
\label{eq:count}  \vert \Spec_{L^2_0}(D_B(\alpha)) \cap \CC/\Gamma^* \vert =2  ,
\end{equation} 
where the elements of the spectrum are included according to their (algebraic) multiplicity. 
In addition, for a fixed constant $ a_0 > 0 $ and for every $ \epsilon $ there exists $ \delta $ such that
for $ 0 < |B | < \delta $, $ |\alpha - \underline \alpha | < a_0 \delta |B| $, 
\begin{equation}
\label{eq:nearG}  \Spec_{ L^2_0 } ( D_B ( \alpha ) ) \subset \Lambda^* + D (0 , \epsilon ) ,
\end{equation}
where we recall that elements of $ \Lambda^* $, in particular $ 0 $, correspond to the $ \Gamma $ point.
\end{theo}

\noindent 
{\bf Remarks.}
1. Existence of the first real magic angle $ \underline \alpha \simeq 0.585 $ was
proved by Watson--Luskin \cite{lawa} and its simplicity (including the simplicity as an 
eigenvalue of the operator $ T_k $ defined in \eqref{eq:defTk}) in \cite{bhz1}, with computer assistance in 
both cases. Numerically,
the simplicity is valid at the computed real elements of $ \mathcal A $ for the Bistritzer--MacDonald potential used
in \cite{magic}.

\noindent
2. The constant $ g_0 ( \alpha ) $ can be evaluated numerically (and its non-vanishing
for the first magic angles could be established via a computer
assisted proof) and here are the results for the (numerically) simple magic angles 
for the potential $ U_{BM}$ in Figure \ref{f:1}:
\medskip

\begin{table}[h!]
\begin{center}
\begin{tabular}{||c|c c c c c c c||} 
 \hline
 Magic angle $\underline{\alpha}$ & 0.585 & 2.221 & 3.751 & 5.276 & 6.794 & 8.312 &9.829  \\ [0.5ex] 
 \hline
 $ \vert g_0(\underline{\alpha})\vert \simeq $ & 7e-02
  & 5
  e-04 & 7
  e-04 & 2
  e-05 & 3
  e-05 &  9
  e-07 & 6
  e-06\\
  \hline
$\vert g_1(\underline \alpha) \vert \simeq$  &   1.3035 & 0.2881 & 0.0880 & 0.0252 & 0.0068 & 0.0017 & 1.7326e-04\\ \hline
\end{tabular}
\end{center}
\caption{Values of $g_0(\underline \alpha)$ and $g_1(\underline \alpha)$ at first magic angles.}
\label{table:2}
\end{table}

\medskip

\noindent
3. The combination of Theorems \ref{t:symD} and \ref{th:Gamma} shows that
for any $ U \Subset (\mathbb C \setminus \mathcal A ) \cup \{ \underline \alpha \} $
(with $ \underline \alpha $ satisfying the assumptions of Theorem \ref{th:Gamma})
there exists $ \delta = \delta ( U ) $ such that $ 0 < |B| < \delta $, the spectrum of
$ D_B ( \alpha ) $ is discrete and $ \vert \Spec_{L^2_0}(D_B(\alpha)) \cap \CC/\Gamma^* \vert =2 $.

\medskip

From the symmetries in \eqref{eq:symD} we conclude that for special values of $ 
\theta = 0, \pm \frac23 $ the spectrum of $ D_B ( \alpha ) $ has a particularly nice
structure as $ \alpha $ varies. We state the result for $ \theta = 0 $, as
we can use the first identity in \eqref{eq:symD} to obtain the other two.
\begin{theo}
\label{th:lin}
\noindent
For $ 0 < B \ll 1 $, 
\begin{equation}
\label{eq:rect1} 
\Spec_{L^2_0} ( D_B  (\alpha ) ) 
\subset  \mathscr R :=   { 2 \pi }   ( i \mathbb R +    \mathbb Z ) \cup \tfrac{2 \pi }{ \sqrt 3} (  \mathbb R +  
  i  \mathbb Z ), \ \ \ \alpha \in \mathbb R \setminus \mathcal A . \end{equation}
Moreover, if the assumptions of Theorem \ref{th:Gamma} are satisfied 
at $ \underline \alpha \in \mathbb R $
then for every $ \epsilon >0 $ there are $\delta_0, \delta_1  >0$ such that 
\begin{equation}
\label{eq:bigcup}\mathscr R \setminus \bigcup_{k \in \mathcal K_0}  B (k  , \epsilon ) \subset \bigcup_{ \underline \alpha - \delta_1 < \alpha < \underline \alpha + \delta_1 } \Spec_{L^2_0} ( D_B  (\alpha ) )  \subset \mathscr R, \ \ \  0 < B < \delta_0. \end{equation}
In addition, for every $k \in \mathscr R \setminus \bigcup_{k \in \mathcal K_0}  B (k  , \epsilon ) $ there is a unique $\alpha \in (\underline \alpha - \delta_1 < \alpha < \underline \alpha + \delta_1)$ such that $k \in \Spec_{L^2_0} ( D_B  (\alpha ) ).$ 
\end{theo}

\noindent
{\bf Remarks.} 1. A more precise statement about the behaviour at $ \mathscr R $ is given in  
Propositions \ref{p:SB} and \ref{p:SB1} -- the implicit formulas for $ \lambda = 1/\alpha $ 
in terms of $ k $ and $ B $ describe a bifurcation phenomenon. {In particular, when 
$ B $ is real, the bifurcation of the eigenvalues of $ D_B ( \alpha )$ at $ 0 $ (at the 
specific value of $ \alpha $) is given by \eqref{eq:morela1}. For the bifurcation at the 
vertices of the boundary of the Brillouin zone, see \eqref{eq:morela3}.}
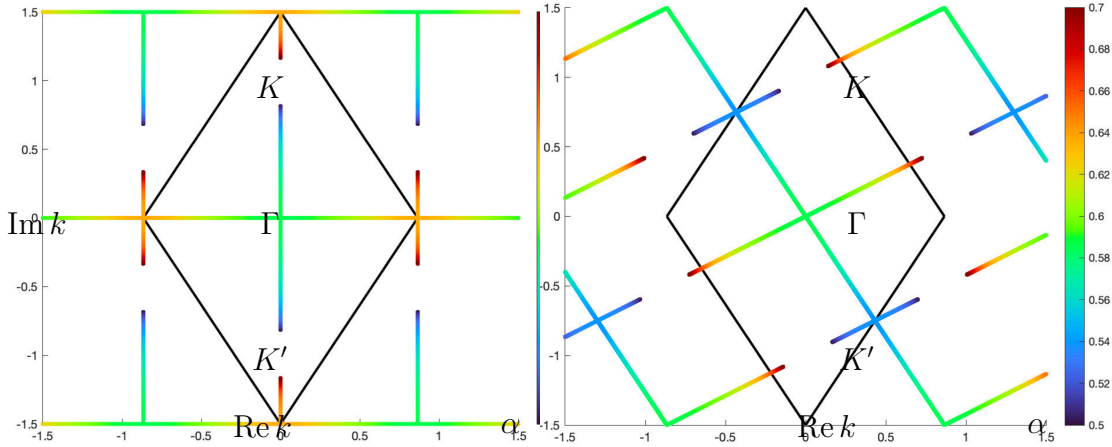
\begin{figure}
{\begin{tikzpicture}
   \node at (0,0) {\includegraphics[width=7.6cm]{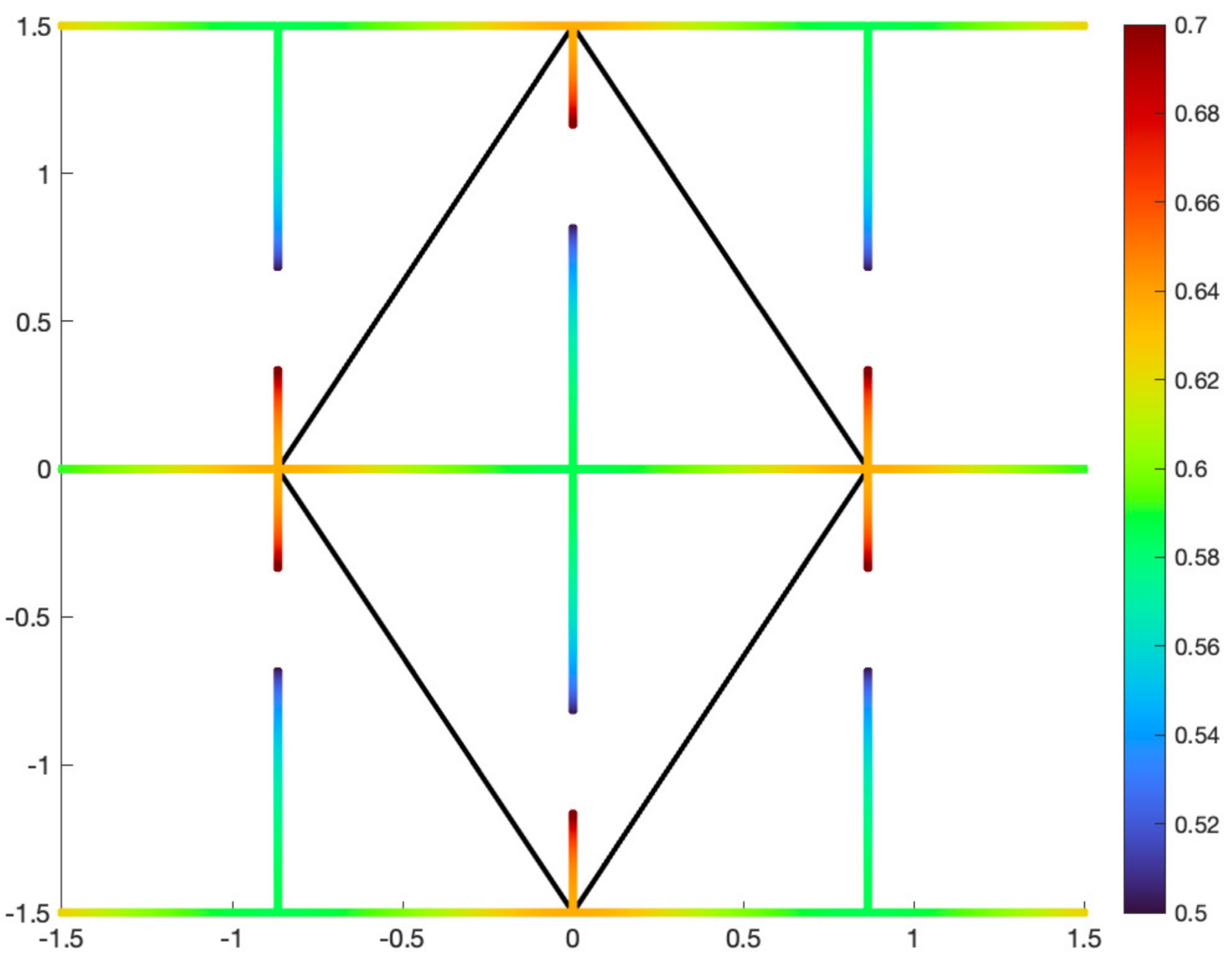}
};  

\node at (7.0,0){\includegraphics[width=7.6cm]{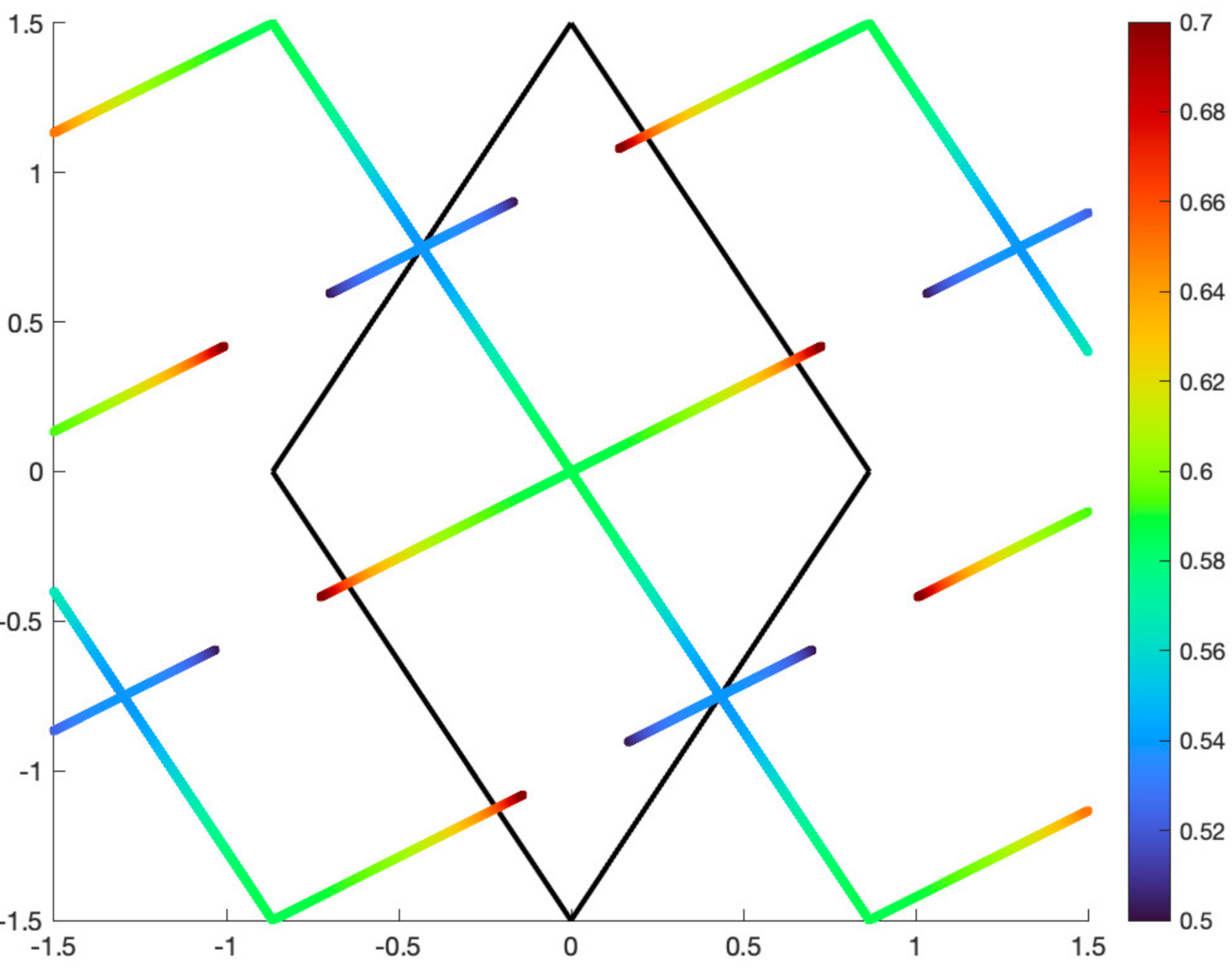}};
\node at (-1+0.5,-2.7) {$\Re k$};
\node at (-3.5,0) {$\Im k$};
\node at (7,-2.7) {$\Re k$};
\node at (2.8,-2.7) {$\alpha$};
\node at (7.6+2.2,-2.7) {$\alpha$};
\node at (-0.4,1.8){$K$};
\node at (-0.4,-1.8){$K'$};
\node at (-0.4,0){$\Gamma$};
\node at (7.4,0){$\Gamma$};
\node at (7.4,1.8){$K$};
\node at (7.4,-1.8){$K'$};
  \end{tikzpicture}} 
\caption{\label{fig:bifurc}Dirac point trajectory for $B=0.1$ (left) and $B=0.1 \omega $ (right). The bifurcation happens at $\Gamma$ and one additional point (modulo $\Lambda^*$) in each figure, respectively. The colors indicate the position of the Dirac cones for given values of $\alpha.$ 
The exclusion of $ K$ and $K' $ points in the statement of Theorem \ref{th:lin} seems to be
a technical issue, as shown in \url{https://math.berkeley.edu/~zworski/Rectangle_2.mp4} 
(for the case of the figure on the right).}
\end{figure}

\noindent
2. The inclusion \eqref{eq:rect1} means that the spectrum lies on a grid of straight lines parallel to the $ x $ and $ y $ axes -- see \url{https://math.berkeley.edu/~zworski/Rectangle_1.mp4}. 
To obtain the sets of other rectangles we
use the the first identity in \eqref{eq:symD}, that is take $ B = \omega b$, $ b > 0 $.

\section{Review of magic angle theory}

We start with a general discussion of operators arising in chiral TBG models.
\subsection{Bloch--Floquet theory}
\label{s:BF}
We recall that 
\[ \Lambda := \mathbb Z \oplus \omega \mathbb Z , \ \ \ \omega := e^{ 2 \pi i/3 } ,
\ \ \   \omega \Lambda = \Lambda , \ \ \ 
  \Lambda^* = \frac {4 \pi i}  {\sqrt 3}  \Lambda . \] 
(The dual basis of $ \{ 1, \omega \} $ is given by $ \{ - 4 \pi i \omega / \sqrt 3, 4 \pi i / \sqrt 3 \} $.) 

We then consider a generalization of \eqref{eq:defDB}:
\[ D ( \alpha ) := 2 D_{\bar z} + \alpha V ( z ) : H^1_{\rm{loc} } ( \mathbb C ; \mathbb C^n ) 
\to L^2_{\rm{loc}} ( \mathbb C ; \mathbb C^n ) , \ \ \ 
H( \alpha ) := \begin{pmatrix} \ \ 0 & D ( \alpha)^* \\
D ( \alpha ) & 0 \end{pmatrix} ,  \]
where $ V ( z ) := C^\infty ( \mathbb C ; \mathbb C^n \otimes \mathbb C^n ) $.
Let $ \rho : \Lambda \to U (n ) $ be a unitary representation and assume that
\begin{equation} 
\label{eq:defV}   V ( z + \gamma ) = \rho( \gamma)^{-1} V ( z ) \rho ( \gamma ) .
\end{equation}
We note that without loss of generality (amounting to a basis change on $ \mathbb C^n$)
 we can assume that
\begin{equation}
\label{eq:rho}  \rho( \gamma ) = {\rm{diag}} \, \left[ ( \chi_{k_j } ( \gamma ) )_{j=1}^n \right] , \ \ \
k_j \in \mathbb C/ \Lambda^* , \ \ \ \ \chi_k ( \gamma ) := \exp ( i \langle \gamma , k \rangle ) . 
\end{equation}

If in the corresponding basis,  $ V ( z ) = ( V_{ij} ( z ) )_{ 0 \leq i, j \leq k } $,  then \eqref{eq:defV} means that
\begin{equation}
\label{eq:Vij}   V_{ij} ( z + \gamma ) = \exp ( i \langle \gamma, k_j - k_i \rangle )  V_{ij} ( z ) . 
\end{equation}
If we define 
\[   \rho ( z ) := {\rm{diag}} \, \left[ ( e^{ i \langle z, k_j \rangle }  )_{j=1}^n \right], \]
then 
\[   V_\rho ( z + \gamma ) = V_\rho ( z ) , \ \ \  V_\rho ( z ) :=  \rho ( z ) V ( z ) \rho ( z )^{-1}  , \]
and
\begin{equation}
\label{eq:Prho}  \rho ( z )  D ( \alpha ) \rho ( z )^{-1}   = D_\rho ( \alpha )  , \ \ \ 
D_\rho ( \alpha ) :=  {\rm{diag}} \, \left[ ( 2D_{\bar z } - k_j )_{j=1}^n \right] + V_\rho ( z ) , \end{equation}
which is a periodic operator. 
In view of this standard Bloch-Floquet theory applies, which can be presented using 
modified translations:
\[ \mathscr L_\gamma u ( z ) := \rho ( \gamma ) u ( z + \gamma ) , \ \ \ 
\mathscr L_\gamma : \mathscr S' ( \mathbb C, \mathbb C^n )  \to \mathscr S' ( \mathbb C, 
\mathbb C^n ) . \]
We have
\[   \mathscr L_\gamma D ( \alpha ) = D ( \alpha ) \mathscr L_{\gamma }. \]
Thus, we can define a generalized Bloch transform
\[ \begin{gathered}    \mathcal  B u ( z, k ) := \sum_{ \gamma \in \Lambda } e^{ i \langle  z + \gamma, k \rangle}   \mathscr L_\gamma  u ( z ) ,  
 \ \ \  \mathcal B u ( z , k + p ) =  
e^{ i \langle z , p \rangle}\mathcal B u  ( z, k )  , \ \ p \in \Lambda^*, \ \ \ u \in \mathscr S ( \mathbb C ) , 
\\
\mathscr L_{\alpha  } \mathcal  B u ( \bullet, k ) = 
\sum_{\gamma} e^{ i \langle z + \alpha + \gamma , k \rangle } \mathscr L_{\alpha + \gamma } 
u ( z ) =  \mathcal B u ( \bullet, k ) , \ \ \alpha  \in \Lambda
\end{gathered} \]
such that (extending the actions of $ \mathscr L_\gamma $ and $ \mathcal B $
to $ \mathbb C^n \times \mathbb C^n $-valued functions diagonally) 
\begin{equation}
\label{eq:conj}    \begin{gathered} \mathcal B D ( \alpha ) = ( D ( \alpha ) - k ) \mathcal B , \ \ \
D ( \alpha ) - k = e^{ i \langle z, k \rangle } D ( \alpha )  e^{ - i \langle z, k \rangle }, \\
\mathcal B H ( \alpha ) = H_k ( \alpha ) \mathcal B , \ \ \ 
H_k ( \alpha )  := e^{ i \langle z, k \rangle } H( \alpha ) e^{ - i \langle z, k \rangle } = 
\begin{pmatrix} \ \ 0 & D ( \alpha)^* - \bar k \\
D ( \alpha ) -k  & 0 \end{pmatrix}  . \end{gathered}  \end{equation}
We check that
\[  \begin{split} \int_{ \mathbb C/\Lambda } \int_{\mathbb C/\Lambda^* } 
| \mathcal B u ( z, k ) |^2 dm ( z ) d m ( k ) & = 
| \mathbb C / \Lambda^* |\int_{ \mathbb C } | u ( z ) |^2 dm ( z ) , 
\end{split} 
\]
and that
\[ \mathcal C v (z ) := | \mathbb C / \Lambda^* |^{-1} \int_{ \mathbb C/ \Lambda^* } v ( z, k ) e^{ - i \langle z, k \rangle} dm ( k ) \]
is the inverse of $ \mathcal B $. 
We now define 
\[   H^s _0  = H^s _0 ( \mathbb C; \mathbb C^k ) := \{ u \in H^s_{\rm{loc} } ( \mathbb C ; \mathbb C^k ) : 
\mathscr L_{\gamma } u = u , \ \gamma \in \Lambda \}  , \ \ 
 L^2 _0  :=  H^0 _0 , \ \ \   k = n, 2n, \]
We have a unitary operator identifying $ L^2 _0 $ with $ L^2 ( \mathbb C / \Lambda ) $
\begin{equation}
\label{eq:symmetry_good}
   \mathcal U_0 u ( z ) := \rho ( z ) u ( z) , \ \ \  \mathcal U_0 : L^2 _0 \to L^2 ( \mathbb C/\Lambda; \mathbb C^n  ) , \ \ \  \mathcal U _0 D(\alpha ) \mathcal U _0^* = D_\rho ( \alpha ) ,  
   \end{equation}
where used the notation of \eqref{eq:Prho}. 

In view of this,
$ \Spec_{ L^2 _0 } ( H_k ( \alpha ) ) $ (with the domain given by $H^1 _0 $) 
is discrete and 
\[ \Spec_{ L^2 ( \mathbb C ; \mathbb C^{2n} ) } ( H ( \alpha ) ) = \bigcup_{ k \in \mathbb C/\Lambda^* } 
\Spec_{ L^2 _0 } H_k ( \alpha ) . \]

Since for $ p \in \Lambda^* $, 
\begin{equation}
\label{eq:deft}  \tau( p ) : L^2_0 \to L^2_0 , \ \ [ \tau ( p ) u ] ( z ) := e^{ i \langle z , p \rangle } u ( z ) , \ \ \
\tau ( p )^{-1} =\tau ( p )^* ,\end{equation}
and 
\[ \tau ( p )^* D ( \alpha ) \tau ( p ) = D ( \alpha ) + p , \]
we have 
\begin{equation}
\label{eq:persp} \Spec_{L^2  _0 }  D ( \alpha ) = \Spec_{L^2  _0 } D ( \alpha ) + \Lambda^* . 
\end{equation}
Finally, we use \eqref{eq:Prho} and 
$ \Spec_{ L^2 ( \mathbb C/\Lambda; \mathbb C ) } ( 2 D_{\bar z } ) = \Lambda^* $ (with simple eigenvalues)
to see that (for $ \rho $ given by \eqref{eq:rho})
\begin{equation}
\label{eq:specD}  \Spec_{ L^2 _0 } ( 2 D_{\bar z }  ) = \bigsqcup_{j=1}^n ( \Lambda^* - k_j ) , \ \ \ 
\text{ Domain of $  2 D_{\bar z }  = H^1 _0 . $} \end{equation}
\subsection{Rotational symmetries}
\label{s:rots} 

We now introduce
\[   \Omega u ( z ) := u ( \omega  z) , \ \  u \in \mathscr S' ( \mathbb C; \mathbb C^n ) , \]
and in addition to \eqref{eq:defV} assume that
\begin{equation}
\label{eq:Vom}    V ( \omega z ) = \omega V ( z ) .\end{equation}
(We do not have many options here as $ \Omega D_{\bar z } = \omega D_{\bar z} \Omega $). Then
\[  \Omega D ( \alpha ) = \omega D ( \alpha ) \Omega , \]
and
\[  \mathscr C H ( \alpha ) = H( \alpha ) \mathscr C , \ \ \ \mathscr C := \begin{pmatrix}
\Omega &  \ 0 \\
 0 & \bar \omega \Omega \end{pmatrix} : \mathscr S' ( \mathbb C ; \mathbb C^n \times
\mathbb C^n ) \to \mathscr S' ( \mathbb C ; \mathbb C^n \times
\mathbb C^n ) . \]
We have the following commutation relation 
\[ \begin{split}    \mathscr L_{\gamma} \Omega u ( z )  & = \rho ( \gamma ) u ( \omega ( z + \gamma ) ) 
= \rho (  \gamma - \omega \gamma   )  \rho ( \omega \gamma ) u ( \omega z + \omega \gamma )
\\
& = \rho ( \gamma - \omega \gamma ) \Omega  \mathscr L_{\omega \gamma } u  ( z ) . 
\end{split} \]
A natural case to consider is given by 
\begin{equation}
\label{eq:rhom}  \rho ( \gamma ) = \rho ( \omega \gamma ) , \ \ \forall \, \gamma \in \Lambda,
\end{equation}
which implies that
\begin{equation}
\label{eq:cube}  \rho ( \gamma )^3 = \rho ( \gamma + \omega \gamma + \omega^2 \gamma  ) = \rho ( 0 )   = I_{\mathbb C^n}  . \end{equation}
In the notation of \eqref{eq:rho}, condition \eqref{eq:rhom} means that
\begin{equation}
\label{eq:defK}   \bar \omega k_j \equiv k_j \!\!\!\!\! \mod \Lambda^* \ \Longleftrightarrow \
k_j  \in \mathcal K :=  
\frac{4 \pi i }{\sqrt 3 } \left( \left\{ 0 , \pm ( \tfrac 1 3 + \tfrac 2 3 \omega ) \right\} + \Lambda \right)  . \end{equation}
We see that $ \mathcal K / \Lambda^* $ is the subgroup of fixed points of multiplication
$ \omega : \mathbb C/\Lambda^* \to \mathbb C/\Lambda^*  $ and it is 
isomorphic to $ \mathbb Z_3 $.

Since \eqref{eq:rhom} implies that 
\[  \mathscr L_\gamma \Omega = \Omega \mathscr L_{\omega \gamma} , \ \ \
\mathscr L_\gamma \mathscr C = \mathscr C \mathscr L_{\omega \gamma}, 
\ \ \  \mathscr C \mathscr L_\gamma  =  \mathscr L_{\bar \omega \gamma} \mathscr C, 
\]
we follow \cite[\S 2.1]{beta} combine the two actions into a group of unitary action which
commute with $ H ( \alpha ) $:
\begin{equation}
\label{eq:defG}  
\begin{gathered}  G := \Lambda \rtimes \ZZ_3 , \ \ 
  \ZZ_3 \ni \ell : \gamma \to \bar \omega^\ell  \gamma , \ \ \ ( \gamma , \ell ) \cdot ( \gamma' , \ell' ) = 
( \gamma + \bar \omega^\ell \gamma' , \ell + \ell' ) ,
\\  ( \gamma, \ell ) \cdot u = 
\mathscr L_{\gamma } \mathscr C^\ell u, \ \ \ u \in L^2_{\rm{loc}} ( \mathbb C ; \mathbb C^n \times 
\mathbb C^n ) . 
\end{gathered}
\end{equation}
By taking a quotient by $ 3 \Lambda $ we obtain a finite group  which acts 
unitarily on $ L^2 ( \CC/ 3\Lambda ) $, and that action commutes with $ H ( \alpha ) $:
\begin{equation}
\label{eq:defG3}
G_3 :=  G/3 \Lambda  = \Lambda /3\Lambda \rtimes \ZZ_3 \simeq \ZZ_3^2 \rtimes \ZZ_3. 
\end{equation}

By restriction to the first two components, $ G $ and $ G_3 $ act on 
$\CC^n  $-valued function and 
use the same notation for those actions.

The key fact
(hence the name \emph{chiral model}) is that
\begin{equation}
\label{eq:defW}
\begin{gathered}   H (\alpha) = - \mathscr W H (\alpha) \mathscr W, \ \ \ \mathscr W := \begin{pmatrix}
1 & 0 \\
0 & -1 \end{pmatrix} : \mathbb C^n \times \mathbb C^n \to  \mathbb C^n \times \mathbb C^n , 
\\  
\ \ \ \mathscr W  \mathscr C  = 
 \mathscr C  \mathscr W , 
 \  
 \ \ \mathscr L_{\gamma } \mathscr W = \mathscr W \mathscr L_{\gamma }.
\end{gathered}
\end{equation}


\subsection{Protected states} 
We now make the assumption \eqref{eq:rhom} and consider the question of protected 
states. We are looking for the set $ \mathcal K_0 \subset \mathbb C $ such that 
\begin{equation}
\label{eq:K0}  \forall \, \alpha \in \mathbb C, \ k \in \mathcal K_0,  \ \  0 \in \Spec_{ L^2 _0 }  H_k ( \alpha ). 
\end{equation}
This condition is equivalent to 
\[  k \in \Spec_{L^2  _0 } D ( \alpha ) 
\ \Longleftrightarrow \ k \in \Spec_{L^2 ( \mathbb C/\Lambda; \mathbb C^n) } D _0 ( \alpha )
, \]
where we used the notation of \eqref{eq:Prho}.  Putting $ \alpha = 0 $ we see
that $ \mathcal K_0  \subset \mathcal K  $.

The following simple lemma is used a lot. To formulate it we introduce the following spaces:
\begin{equation}
\label{eq:defHk}  
H^s_{k}  := \{ u \in H^s ( \mathbb C/ 3 \Lambda ; \mathbb C^2 \times \mathbb C^2 ): 
\mathscr L_\gamma u = e^{ i \langle k , \gamma \rangle } 
u \}, \ \ 
  k \in \mathcal K/ \Lambda^*\simeq \mathbb Z^3, \ \ 
p \in \mathbb Z^3 , \end{equation} 
(with the corresponding definition of $ L^2_{k} $). 

\begin{lemm}
\label{l:ev2ker}
Suppose that $ k, k'  \in \mathcal K $ and $ \tau ( k ) $ is defined as in \eqref{eq:deft}. Then
in the notation of \eqref{eq:defHk}, 
$   \tau ( k ) : H^s_{ k' } \to H^s_{k' + k } $ and 
\begin{equation}
\label{eq:ev2ker}  
\begin{gathered}   \tau ( k ) :  \ker_{ H^1_{0} } ( D ( \alpha ) + k ) = \ker_{ H^1_{ k} } D ( \alpha ) , \\
\tau (k ) :  \ker_{ H^1_ 0  } H_{-k} ( \alpha ) = \ker_{ H^1_{k} } H ( \alpha ) . 
\end{gathered}
\end{equation}
\end{lemm}
\begin{proof}
We have $ \tau ( k ) = e^{ i \langle k , z \rangle } $ (as a multiplication operator) and for 
$ u \in H^s_{k'} $, 
\[ \mathscr L_\gamma ( \tau ( k ) u ) ( z ) = e^{ i \langle k , z +\gamma \rangle } \mathscr L_\gamma u ( z ) 
= e^{ i \langle k +k' , \gamma \rangle } \tau ( k ) u ( z ) , \]
which proves the mapping property of $ \tau ( k ) $. 
Also, 
$ D ( \alpha) w = e^{ i \langle z , k \rangle} ( D  ( \alpha ) + k ) ( 
e^{- i \langle z, k \rangle } w ) $. Hence if  $  (D ( \alpha ) + k )  u = 0 $ and $ 
\mathscr L _\gamma u = u $ then
$ w := e^{  i \langle z, k \rangle } u \in H^1 ( \mathbb C/3 \Lambda; \mathbb C^{2n} ) $, 
$ D ( \alpha ) w = 0 $, and $ \mathscr L_\gamma w = \mathscr L_\gamma ( e^{  i \langle z, k \rangle } u)
= e^{  i  \langle z + \gamma , k \rangle } \mathscr L_\gamma u = e^{ i \langle \gamma , k \rangle } w $,
that is $ w \in H_k^1 $. 
\end{proof}

We are interested in the case of $ n = 2 $ and obtain the following reinterpretation of 
earlier protected states statements -- see \cite{magic}.
\begin{prop}
\label{p:prot}
If $ n = 2 $ (in the notation of \eqref{eq:rho} and \eqref{eq:K0})
and $ k_1 \not \equiv k_2 \!\!\! \mod \Lambda^*  $, 
$ k_j \in \mathcal K $,  then $ \mathcal K_0 = \{- k_1,- k_2 \} + \Lambda^* $.
\end{prop}
\begin{proof} 
We use \eqref{eq:ev2ker} and decompose $ \ker_{H^1 ( \mathbb C / 3 \Lambda; \mathbb C^{4} )}
H ( \alpha) $ into representations of $ G_3 $ given by \eqref{eq:defG}. 
From \eqref{eq:defW} we see that the spectrum of $ H( \alpha ) $ restricted to a representation of
$ G_3 $ is symmetric with respect to the origin. If (see \cite[\S 2.2]{beta} for a review
of representations of $ G_3 $)
\begin{equation}
\label{eq:defHkp}  H^s_{ { k , p } } := \{ u 
\in  H^s ( \mathbb C/ 3 \Lambda ; \mathbb C^2 \times \mathbb C^2 ): 
\mathscr L_\gamma \mathscr C^\ell  u = e^{ i \langle k , \gamma \rangle } 
\bar  \omega^{\ell p} u \},  
\end{equation}
$ k \in \mathcal K/ \Lambda^*\simeq \mathbb Z^3$, $ 
p \in \mathbb Z^3 $,  (with the corresponding definition of $ L^2_{ {k,p} } $)
then the constant
functions (given by the standard basis vectors in $ \mathbb C^4 $) satisfy
\[  \mathbf e_1 \in H^1_{ {k_1, 0} }, \ \  \mathbf e_2 \in H^1_{{k_2, 0 } }, \ \ 
\mathbf e_3 \in H^1_{{k_1, 1} } , \ \  \mathbf e_4 \in H^1_{{k_2, 1 } }, \]
and since $ k_1 \not \equiv k_2 \!\! \mod \Lambda^* $, all these spaces are different. 
The spectrum of $ H ( \alpha ) |_{  L^2_{ {k,p} } } $ is even (see \eqref{eq:defW})
and $  \ker_{ H^1_{  { k_j, p }}} H( 0 ) = \mathbb C \mathbf e_{j+2p} $, $j = 1,2$, 
$ p = 0,1 $. Continuity of eigenvalues shows that
\begin{equation}
\label{eq:dimker}     \dim \ker_{ L^2_{ k_j, p } } H( \alpha ) \geq 1, \ \ \ \alpha \in \mathbb C , \ \ \ 
j=1,2 ,  \ \ p = 0,1 . \end{equation}
which in view of Lemma \ref{l:ev2ker} concludes the proof.
\end{proof} 

\noindent
{\bf Remark.} Under the assumptions of Proposition \ref{p:prot} the
corresponding $ - k_1 ,- k_2 \in \mathbb C/ \Lambda^* $ are called the $ K $ and $ K' $ points
in the physics literature. The remaining element of $ \mathcal K / \Lambda^* $ is 
called the $ \Gamma $ point. 

Existence of protected shows that we have a natural labelling for the eigenvalues of 
$ H( k ) $ on $ L^2 _0 $:
    \begin{equation}
\label{eq:eigs} 
\begin{gathered} \Spec_{ L^2 _0 } ( H( k ) ) = 
 \{ E_{ j } ( \alpha, k ) \}_{ j \in  \mathbb Z^* } ,  \ \ \  E_{j } ( \alpha, k ) 
= - E_{-j} ( \alpha , k ) , \\ 0 \leq E_1 ( \alpha,  k ) \leq E_2 ( \alpha, k ) \leq \cdots , \ \ \  E_{\pm 1} ( \alpha, -k_1 ) =  E_{\pm 1} ( \alpha, -k_2  ) = 0 . 
\end{gathered} \end{equation}
where the eigenvalues are included according to their multiplicities (and $ \mathbb Z^* := 
\mathbb Z \setminus \{ 0 \} $).

\subsection{Magic angles}
\label{s:magic}

We recall the main result of \cite{beta}, the spectral characterization of
{\em magic angles}. See also proof of \cite[Proposition 2.2]{bhz2}. 

\begin{prop}
\label{p:magicD}
Suppose that $ n = 2 $ and that the condition \eqref{eq:rhom} holds. 
Then, in the notation of Proposition \ref{p:prot} there exists a discrete set $ \mathcal A$ such that
\begin{equation}
\label{eq:magic} 
\Spec_{ L^2 _0 } D ( \alpha ) = \left\{ \begin{array}{ll}  \mathcal K_0  & \alpha \notin 
\mathcal A , \\
 \mathbb C & \alpha \in \mathcal A . \end{array} \right.
\end{equation} 
Moreover, 
\begin{equation}
\label{eq:specc}  \alpha \in \mathcal A \ \Longleftrightarrow \ \exists \, k \notin \mathcal K_0, \ 
\alpha^{-1} \in \Spec_{ L^2 _0 } T_k \ \Longleftrightarrow \forall \, k \notin \mathcal K_0, \ 
\alpha^{-1} \in \Spec_{ L^2 _0 } T_k, 
\end{equation}
where $ T_k $ is a compact operator given by 
\begin{equation}
\label{eq:defTk}   T_k := R ( k )  V ( z ) : L^2 _0 \to L^2 _0 , \ \ \ R( k ) :=  ( 2 D_{\bar z } - k )^{-1}  
\end{equation}
\end{prop}

\subsection{Antilinear symmetry} 
\label{s:anti}

We will make the following assumption:
\begin{equation}
\label{eq:defAsc}  \mathscr A D ( \alpha )  = -   D ( \alpha )^* \mathscr A , \ \ 
\mathscr A :=  \begin{pmatrix} \ \   0 & \Gamma  \\ - \Gamma & 0  \end{pmatrix}, \ \ 
\Gamma v ( z ) = \overline{ v ( z ) } .
 \end{equation}
 A calculation based on the definition of $ \mathscr L_\gamma $ gives
  \begin{equation}
 \label{eq:mapA} \mathscr A : L^2_{ k, p } \to L^2_{ - k + k_1 + k_2 ,-  p } , \ \ \ 
 k \in \mathcal K , \ \ p \in \mathbb Z_3. 
  \end{equation}
 In particular if (as we assume) $ k_1 \not \equiv k_2 \!\! \mod \Lambda^* $ and
 $ k_0 \notin \{ k_1, k_2 \} + \Lambda^* $, then 
 $ -k_0 + k_1 +k_2 
 \equiv  k_0 \!\! \mod \Lambda^* $, and consequently 
 \begin{equation}
 \label{eq:mapA1} \mathscr A : L^2_{ k_0, p } \to L^2_{ k_0 ,-  p } , \ \ \ 
 p \in \mathbb Z_3. 
  \end{equation}

Since (we put $ \alpha = 1 $ to streamline notation; that amounts to absorbing 
$ \alpha $ into $ V $)
\[  \mathscr A \begin{pmatrix} V_{11} & 0 \\
0 & V_{22} \end{pmatrix} 
=  - \begin{pmatrix} - \bar V_{22}  & \ 0 \\ 0 & - \bar V_{11}    \end{pmatrix}\mathscr A , \]
for \eqref{eq:defAsc} to hold we need $ V_{11} = - V_{22} =: W_1 $. From \eqref{eq:Vij} we see
that $ W_1 $ is $ \Lambda$-periodic and there exists $ \Lambda$-periodic $ W_0 $ such that
\begin{gather*}  
\begin{pmatrix} 2D_{\bar z } + W_1 & 0 \\
0 & 2D_{\bar z } - W_1 \end{pmatrix} = 
\begin{pmatrix} e^{W_0 ( z ) } & 0 \\ 0 & e^{-W_0 ( z ) } \end{pmatrix}
\begin{pmatrix} 2D_{\bar z }  & 0 \\
0 & 2D_{\bar z } \end{pmatrix}
\begin{pmatrix} e^{- W_0 ( z ) } & 0 \\ 0 & e^{ W_0 ( z )} \end{pmatrix} , \\
2 D_{\bar z } W_0 = W_1 , \ \ \
W_0 ( \omega z ) = W_0 ( z ) .
 \end{gather*}
(From \eqref{eq:Vom} we see that $ W_1 ( \omega z ) = \omega W_1 ( z ) $ and 
hence the integral of $ W_1 $ over $ \mathbb C/\Lambda $ is equal to $ 0 $; this shows
that we can find $ W_0 $, which is unique up to an additive constant.) We conclude that
if we insist on \eqref{eq:defAsc} then we can, without loss of generality assume that
\begin{equation}
\label{eq:genV}  
\begin{gathered}   V ( z ) = \begin{pmatrix}  \ 0 & V_{12} ( z) \\
V_{21} ( z ) & \ 0 \end{pmatrix}, \\ V_{ i j } ( z + \gamma ) = e^{ i \langle k_j - k_i , \gamma \rangle } V( z ) , \ \ 
k_\ell \in \mathcal K, \ \ k_1 \neq k_2 , \ \ V_{ij } ( \omega z ) = \omega V_{ij} ( z ) .
\end{gathered}
\end{equation}

To verify the latter we check that, with $ w = (w_1 , w_2 )  $,  
\[ \begin{split}   \begin{pmatrix} 2 D_{\bar z} & V_{12} \\
V_{22} & 2 D_{\bar z }  \end{pmatrix} \mathscr A  w & = \begin{pmatrix} 2 D_{\bar z } \Gamma w_2 - 
 V_{12} \Gamma w_1 \\ 
- 2 D_{\bar z } \Gamma w_1 +  V_{21} \Gamma w_2  \end{pmatrix} 
= \begin{pmatrix} \Gamma \left( - 2 D_{ z } w_2 -  \bar V_{12} w_1\right) \\ 
\Gamma \left(  2 D_{ z } w_1 +   \bar V_{21}  w_2 \right)  \end{pmatrix} 
\\ & =- \begin{pmatrix} \ \ 0 &  \Gamma \\
 - \Gamma &  0 \end{pmatrix} \begin{pmatrix} 2 D_{ z } w_1 +  \bar V_{21}  w_2  \\
2 D_{ z } w_2 +   \bar V_{12} w_1 \end{pmatrix}  = 
-  \mathscr A   \begin{pmatrix} 2 D_{\bar z} & V_{12} \\
V_{22} & 2 D_{\bar z }  \end{pmatrix}^* w .
\end{split} 
\]

\noindent
{\bf Remarks.} 1. The antilinear symmetry is closely related to the $ C_{2z} T $ symmetry 
in the physics literature. 

\noindent
2. In the case when $ V_{21} ( z ) = V_{12} (- z ) $, 
we have another antilinear symmetry:
\begin{equation}
\label{eq:H2Q}  
\begin{gathered}
Q v ( z ) := - \mathscr A \mathscr E v ( z ) = \overline{ v ( - z ) } , \ \ \ 
Q D( \alpha ) Q =  D (\alpha )^* .
 \end{gathered}  \end{equation}
 The mapping property is simpler than \eqref{eq:mapA}: $ Q : L^2_{ k, p } ( \mathbb C/\Lambda ; \mathbb C^2 ) \to 
 L^2_{ k, -p } ( \mathbb C/\Lambda ; \mathbb C^2 ) $.

\renewcommand\thefootnote{\dag}%

\subsection{Theta functions}
\label{s:theta} 
We now review properties of theta functions. To simplify notation we put
$ \theta ( z ) :=  \theta_{1} ( z | \omega ) := - \theta_{\frac12,\frac12} ( z | \omega ) $, 
and recall that 
\begin{equation}
\label{eq:theta}
\begin{gathered} 
\theta ( z )
= - \sum_{ n \in \mathbb Z } \exp ( \pi i (n+\tfrac12) ^2 \omega+ 2 \pi i ( n + \tfrac12 ) (z + \tfrac
12 )  ) , \ \ \ \theta ( - z ) = - \theta ( z ) \\
\theta  ( z + m ) = (-1)^m \theta   ( z ) , \ \ \theta  ( z + n \omega ) = 
(-1)^n e^{ - \pi i n^2 \omega - 2 \pi i z  n } \theta   ( z ) ,
\end{gathered}
\end{equation}
and that $ \theta $ has simple zeros at $ \Lambda $ (and no other zeros) -- 
see \cite{tata}. 

We now define 
\begin{equation}
\label{eq:defFk}  F_k ( z ) = e^{\frac i 2   (  z -  \bar z ) k } \frac{ \theta ( z - z ( k ) ) }{
\theta( z) } ,  \ \ \ z (k):=  \frac{ \sqrt 3 k }{ 4 \pi i } , \ \  z:  \Lambda^* \to \Lambda . 
\end{equation}
Then, using \eqref{eq:theta} and differentiating in the sense of distributions, 
\begin{equation}
\label{eq:propFk}
\begin{gathered}F_k ( z + m + n \omega ) = e^{  -  n k \Im \omega } 
e^{ 2 \pi i n z ( k ) } F_k ( z ) = F_k ( z ) , \\
  ( 2 D_{\bar z } + k ) F_k ( z ) 
 = c(k)   \delta_0 ( z ) , \ \ c(k) :=  2 \pi i {\theta ( z ( k ) ) }/{ \theta' ( 0 ) } .
\end{gathered} \end{equation}
(Here we used the fact that if $ f $ and $ g $ are
 holomorphic, $ g ( \zeta ) $ has a simple zero at $ 0 $ 
 and $ f ( 0 ) \neq 0 $ then, near $ 0 $,  $ \partial_{\bar \zeta } ( f ( \zeta ) / g ( \zeta ) )=
 \pi f ( 0)/g'( 0) ) \delta_0 ( \zeta ) $ -- see for instance \cite[(3.1.12)]{H1}.)

The following Lemma is now immediate. It reinterprets the theta function argument
in \cite{magic}.
\begin{lemm}
\label{l:theta}
Suppose that $ p \in \mathcal K $ and $ u \in \ker_{ H^1_{p} } ( D ( \alpha ) + k ) $. 
Then 
\begin{equation}
\label{eq:ltheta} ( D( \alpha ) + k + k' ) ( F_{k'} ( z - z(k'') ) ) u ( z) ) = c_{k' - k''} \delta_{z(k'')} 
( z ) u ( z ( k'' )), \ \ \ k, k', k'' \in \mathbb C , \end{equation}
where $ c ( k ) $ is given in \eqref{eq:propFk}.
In particular, if $ u ( z ( k'' ) ) = 0 $ then 
\begin{equation}
\label{eq:ker2ker}   F_{k'} ( z - z (k'') ) u ( z ) \in \ker_{ H^1_{p } } ( D( \alpha ) + k + k' ) .
\end{equation}
\end{lemm}

\subsection{Multiplicity one}
\label{s:mult} 

The definition of the set of {\em magic} $ \alpha$'s based on 
Proposition \ref{p:magicD}
does {\em not} involve the notion of multiplicity. 
Here we will discuss the case of multiplicity one\footnote{A more general discussion 
is  presented in \cite{bhz3} -- generic simplicity 
presented there is modified in view of protected multiplicity two magic
angles -- see the proof of Proposition \ref{p:mult1}.}. 
One natural definition of multiplicity of 
magic angles is given in terms of eigenvalues of $ H_k ( \alpha ) $ in 
\eqref{eq:eigs}. We first note that
\begin{equation}
\label{eq:A2E}  \alpha \in \mathcal A \ \Longleftrightarrow \  \forall \, k \in \mathbb C / \Lambda^* , \  \ 
E_{\pm 1 } ( \alpha, k ) = 0 . \end{equation}
We then say, that the magic angle $ \alpha \in \mathcal A $ is simple/has multiplicity one, 
if and only if
\begin{equation}
\label{eq:def1}   \forall \, k \in \mathbb C  , \  j > 1, \ \ 
E_{j } ( \alpha, k ) > 0.  \end{equation}
As stated in \eqref{eq:strongm} we use a stronger definition in this paper.

The operators
\[ \mathbb C^2 \ni  ( \alpha , k )  \longmapsto 
D ( \alpha ) + k  : H^1 _0 \to L^2 _0 , \]
 form a continuous family
of Fredholm operators  of index is $ 0 $. (This follows
from the ellipticity of $ D ( \alpha ) $, the continuity of the index 
and then fact \eqref{eq:A2E} implies that $ D ( \alpha ) - k $ is invertible for
some $k $ and $ \alpha $.)
In particular, $ \dim \ker ( D ( \alpha ) + k ) = \dim {\rm{coker }} ( D ( \alpha ) ^* + \bar k ) 
= \dim \ker ( D ( \alpha )^* - \bar k ) $, and hence 
\begin{equation}
\label{eq:defmul1}  
\text{\eqref{eq:def1}} \ \Longleftrightarrow \ \forall \, k \in \mathbb C , \ 
 \dim \ker_{ H^1_0 }  ( D( \alpha ) + k ) = 1. 
  \end{equation}

In \cite[Theorem 2]{bhz2} we proved that
\begin{prop}
\label{p:bhz2} 
Suppose \eqref{eq:genV} holds and that 
\[  k_0 \in  \mathcal K
\setminus \{ k_1, k_2 \} , \ \  k_1 \not \equiv k_2 \! \!\!  \mod  \Lambda^* . \]
Then for $ \alpha \in \mathcal A $ we have
\[ \text{\eqref{eq:def1}} 
 \ \Longleftrightarrow \ \exists \, k \not \equiv k_1, k_2  \! \! \!\! \mod \Lambda^*, \ \ 
 \dim \ker_{ H^1 _0 } ( D ( \alpha ) + k )  = 1 . \]
 In particular, $ \alpha \in \mathbb C $ is a simple magic angle (in the sense of
 \eqref{eq:def1}) if and only if 
\begin{equation}
\label{eq:th7}     \dim \ker_{ H^1 _0 } ( D ( \alpha ) + k_0 )  = 1 .
\end{equation}
\end{prop}
We recall that the proof is based on Proposition \ref{p:magicD} and theta function arguments
reviewed in \S \ref{s:theta}. 

A  symmetric choice of $\rho$ in \eqref{eq:rho} is given by:
\begin{equation}
\label{eq:KKG}    k_1 =  \frac{4  \pi } {i \sqrt 3 }  ( \tfrac 1 3 + \tfrac 2 3 \omega ) =  \tfrac{4 } 3 \pi =: K
, \ \ \ k_2 = 
- K   = \tfrac{ 4  } 3  \pi, 
  \ \ \ k_0 =  0. 
\end{equation}
This corresponds to $ \Gamma = 0 $ in the physics notation. In \cite{beta} we 
followed \cite{magic} and used a non-symmetric (equivalent) choice. This corresponds
to the assumptions in \eqref{eq:defU} with $ k_1 = K  $.

\begin{prop}
\label{p:mult1}
Suppose that  \eqref{eq:th7}  holds.
Then, in the notation of Lemma \ref{l:ev2ker}, 
\begin{equation}
\label{eq:symu0}
\ker_{ H^1 _0 }  ( D ( \alpha ) + k_0 )  
= \mathbb C \tau (k_0 )^* u_0, \ \  \| u_0 \|_{ L^2_{k_0}  } = 1 ,  \ \ 
\Omega u_0 =  \omega u_0 , 
\end{equation}
that is, in the notation of \eqref{eq:defHkp}, $ u_0 \in L^2_{k_0,2} $. In addition, 
\begin{equation}
\label{eq:vanu0}  u_0 ( z ) = z w ( z ) , \ \ w \in C^\infty ( \mathbb C ; \mathbb C^2 ) , \ \  w( 0 ) \neq 0, \ \ \ 
  u_0 ( z ) \neq 0, \ \  z \notin \Lambda . \end{equation}
\end{prop}

\noindent
{\bf Remark.} The key insight in \cite{magic} was to use vanishing of $ u \in
\ker_{H^1_{k_1} }  D ( \alpha )  $ for magic $ \alpha$'s at a distinguished point $ z_S $ to show that 
$ \Spec_{ H^1_0 } ( D ( \alpha ) ) = \mathbb C $. In \cite[Theorems 1]{beta} this was shown equivalent to the spectral 
definition based on Proposition \ref{p:magicD}. Here we take a direct approach: only at magic $ \alpha $'s
we have $  \ker_{H^1_{k_0} }  D ( \alpha )  \neq \{ 0 \} $ and \eqref{eq:symu0} shows that
its elements  have to vanish at $ 0 $. 
Lemma \ref{eq:ltheta} then implies vanishing of other
eigenfunctions. 

\begin{proof}[Proof of Proposition \ref{p:mult1}]
From Lemma \ref{l:ev2ker} and \eqref{eq:th7} we conclude that $ \ker_{ H^1_{k_0} } D ( \alpha ) = 
\mathbb C u_0 $ and as 
$ L^2_{k_0} = \bigoplus_{ j=0}^2 L^2_{k_0, j } $ we can decompose the
kernel using these subspaces. 
Since $  D ( 0 ) + k_0 :H_0^1 \to L^2_0  $ is invertible (see \eqref{eq:specD}), 
 \eqref{eq:ev2ker}  shows that $ D ( 0 ) : H_{k_0 }^1 \to L^2_{k_0 } $ is invertible with the 
inverse given by $ R ( 0 ) $. It then follows 
that (see  \eqref{eq:defTk}) 
\[ I + \alpha T_0 = R ( 0 ) D ( \alpha ) : L^2_{k_0,j} \to L^2_{k_0, j } , \ \ \
\ker_{ H^1_{k_0,j} }D ( \alpha )  = \ker_{ L^2_{k_0,j} } R ( 0 ) D ( \alpha ) . \]
(We do use ellipticity of $ D ( \alpha) $ here: the element of the kernel on $ L^2 $ must automatically be smooth.)
Hence if $ \ker_{ L^2_{k_0,j} } ( R ( 0 ) D ( \alpha ) ) \neq \{ 0 \} $, $ j = 0, 1 $,  then 
 $ \ker_{ L^2_{k_0,j} }(  D (\alpha)^* R(0)^*) \neq \{ 0 \} $, and there exists $ w \in L^2_{k_0, j} $
such that 
$   D ( \alpha )^* R ( 0 )^* w = 0 $. 
Since $ R ( 0 )^* : L^2_{ k_0, j }  \to L^2_{k_0,j-1} $ 
and $ \mathscr A : L^2_{ k_0, j-1 } \to L^2_{k_0, -j+1 } $
(see \eqref{eq:mapA1}), we have 
\[  D ( \alpha ) \mathscr A R ( 0 )^* w = 0 , \ \   \mathscr A R(0)^* w \in L^2_{k_0, -j +1  } \neq L^2_{ k_0, j } 
\ \text{ when $ j = 0 , 1 $.} \]
 This means that $ \dim \ker_{H^1_{k_0} } D( \alpha) > 1 $, contradicting the simplicity assumption.
The simplicity and uniqueness of the zero of $ u_0 $ \eqref{eq:vanu0} 
 follows from \cite[Theorem 3]{bhz2}.
\end{proof}


For an $ \alpha \in \mathcal A $, we assume that \eqref{eq:th7} holds. In that case
Proposition \ref{p:mult1} and Lemma \ref{l:theta} show that
\begin{equation}
\label{eq:u02uk}    \ker_{ H^1_0 }  ( D ( \alpha ) + k ) = \mathbb  C u ( k ) , \ \ \ 
u ( k ) := \frac{ F_k u_0 }{ \| F_k  u_0 \| } . \end{equation}
Using \eqref{eq:defAsc} we see that (since $ \mathscr A^2 = - I $)
\[  ( D ( \alpha )^* + \bar k ) \mathscr A = - \mathscr A ( D ( \alpha ) - k ) , \]
which implies that
\begin{equation}
\label{eq:u02u*}
\ker_{ H^1_0 }  ( D ( \alpha )^* + \bar k ) = \mathbb  C \mathscr A u ( - k ) . 
\end{equation}

\noindent
{\bf Remark.} From \cite[(6.6)]{bhz2} we see that (note the difference of notation: $ u ( k ) 
$ there is not normalized) for the basis of $ \Lambda^* $ satisfying $ z ( e_1 ) = 1 $, 
$  z (e_2 ) = \omega $, we have for $ p =  m e_1 + n e_2 \in \Lambda^* $, 
\begin{equation}
\label{eq:k+p} 
u ( k + p ) = e_p ( k )^{-1} \tau ( p ) u ( k ) , \ \ \ 
e_p ( k ) := e^{ - \frac 12 \pi i n^2  +  \pi i ( k + \bar k)  n } ( -1 )^{n+m} ,
\end{equation}
where the unitary operator $ \tau( p ) $ was defined in \eqref{eq:deft}.

\section{Grushin problems}
\label{s:proof}

In this section we construct Grushin problems (see \cite{SZ} and \cite[\S 6]{notes}) which allow us to treat
small in-plane magnetic fields as perturbations. In \S \ref{s:bif} we combine that with the spectral characterization
of magic angles (Proposition \ref{p:magicD}) to analyze the behaviour at the $ \Gamma $ point
and at the vertices of the boundary of the Brillouin zone.

\subsection{Grushin problem for $ D_B ( \alpha ) $}
Suppose $ \underline \alpha \in \mathcal A $ is simple, in the sense that \eqref{eq:th7} holds.
We then put, in the notation of \eqref{eq:u02uk} and \eqref{eq:u02u*}, 
\begin{equation}
\label{eq:Grushin1} \begin{gathered}
\mathcal D_B  ( \underline \alpha , k ) = \begin{pmatrix} D ( \underline \alpha ) + k & R_- ( k ) \\
R_+ ( k ) & 0 \end{pmatrix} + \begin{pmatrix} \mathcal B & 0 \\ 0 & 0 
\end{pmatrix}  : H^1_0 \times \mathbb C \to L^2_0 \times \mathbb C , 
 \\   R_- (k ) u_- = u^* ( k ) u_- , \ \ \ R_+ ( k ) u = \langle u , u(k) \rangle, \\
( D ( \underline \alpha ) + k ) u( k ) = 0 , \ \ \  \| u ( k ) \| = 1, \ \ 
 u^* ( k ) = \mathscr A u ( -  k ) . \end{gathered}
\end{equation} 
We have 
\[ \mathcal D_B( \underline \alpha, k )^{-1} = \begin{pmatrix}
E^B ( k ) & E^B_+ (  k ) \\
E^B_- (  k ) & E_{-+}^B (  k )   \end{pmatrix}, \]
where
\begin{equation}
\label{eq:defE} \begin{gathered} E_+^0 v_+ := u ( k ) v_+ , \ \ \ E^0_- v := \langle v , u^* ( k ) \rangle, 
\ \ \ E_{-+}^0 = 0 , 
\\ 
E^0 v := \left( ( D ( \underline \alpha ) + k )|_{ ( \mathbb C u( k )) ^\perp \to ( \mathbb C u^*(k))^\perp} \right)^{-1}
( v - \langle v , u^* ( k ) \rangle u^* ( k ) ) . \end{gathered} \end{equation}
We now apply \cite[Proposition 2.12]{notes} to obtain
\begin{equation}
\label{eq:EBpm} 
\begin{gathered}  E^B_{-+} = - E_- \mathcal B E_+  + \mathcal O ( B^2 ) =
- c(k ) c^* ( k ) B (  G ( k )     + \mathcal O ( B  )   ) , \\
G ( k ) := (c(k ) c^* ( k ) )^{-1} \left( \langle u_1 ( k ) , u_1^* ( k ) \rangle - \langle u_2 ( k ) , u_2^* ( k ) \rangle \right) ,
\end{gathered}
\end{equation}
and, if  $ u_0  = ( \psi, \varphi )^t $, and $ u ( k) = ( u_1 ( k ), u_2 ( k ) )^t $ then 
\[ u_1 ( k ) = c ( k )   F_k \psi , \ \ u_2 ( k ) =  c ( k )  F_k \varphi , \ \ 
u_1^* ( k ) = c^*(k)  \overline { F_{- k } \varphi} , \ \ u_2^* = - c^*( k )  \overline{ F_{-k} \psi } , \]
where $ c(k), c^*(k) > 0 $ come from $ L^2$-normalizations of $ u $ and $ u^* $.

Hence,
\[ G ( k ) = 2 \int_{\mathbb C/\Lambda } F_k ( z ) F_{-k} ( z ) \varphi ( z ) \psi ( z ) dm ( z) . \]
In fact $ G ( k ) $ is a multiple of  $ \theta( z (k ) )^2 $ which follows from a theta function identity (see \cite[(4.7a)]{voca} or \cite[\S I.5, (A.3)]{tata}):
\begin{equation}
\label{eq:Weier} \theta ( z + u ) \theta ( z - u ) \theta_2 ( 0 )^2= \theta^2 ( z ) \theta_2^2 ( u ) - \theta_2^2 ( z ) \theta^2 ( u) , \ \ \  \theta_2 ( z ) := \theta ( z + \tfrac12 ) .  \end{equation}
Since, (from $ u \in H^1_{0,2} $) 
\[   \int_{ \mathbb C/\Lambda } \varphi ( z ) \psi ( z ) dm ( z ) = 
 \int_{ \mathbb C/\Lambda } \varphi ( \omega z ) \psi ( \omega z ) dm ( z )  = 
 \omega^2 \int_{ \mathbb C/\Lambda } \varphi ( z ) \psi ( z ) dm ( z ) , \]
this integral vanishes, and \eqref{eq:Weier} gives
\begin{equation}
\label{eq:G2g}  G( k ) = g_0 \, \frac{ \theta ( z ( k ) )^2}{ \theta ( \tfrac12 ) ^2 } , \ \ \ 
g_0  =  g_0 ( \underline \alpha )  := 
2 \int_{\mathbb C/\Lambda }   \ \theta ( z + \tfrac12 )^2  \frac{ \varphi ( z ) \psi ( z )}{ \theta ( z ) ^2 }    dm ( z).
\end{equation}
Numerical evidence, see Table \ref{table:2}, suggests that for the Bistritzer--MacDonald potential and the first magic
angle, 
\[ |g_0| \simeq 0.07 \neq 0 . \]
(The number $ g_0 $ is determined up to phase which we can choose arbitrarily by 
modifying $ u_0 \mapsto e^{ i \theta } u_0 $.) Table \ref{table:2} shows approximate values of $ |g_0| $ for higher magic angles for the same potential.

\noindent
{\bf Remark.} We also see that the Grushin problem \eqref{eq:Grushin1} remains well posed
with $ \underline \alpha $ replaced with $ \alpha $, $ | \underline \alpha - \alpha | \ll 1$. 
The effective Hamiltonian \eqref{eq:EBpm} has to be modified by term (obtained
again using \cite[Proposition 2.12]{notes}) 
\begin{equation}
\label{eq:newE-+} 
\begin{gathered}   E_{-+}^B ( k ,\alpha  ) = E_{-+}^B ( k )  - ( \alpha - \bar \alpha) f_2  ( k, B, \alpha ) , \\ 
f_2 ( k , 0 , \underline \alpha ) := g_1(k,\underline \alpha) = - E_-^0 ( k ) \begin{pmatrix}
\ 0 & U ( z ) \\
U ( -  z ) & 0 \end{pmatrix} E_+^0 ( k ) ,
\end{gathered}
\end{equation} 
where, $g_1(\alpha):=g_1(0,\alpha)$ and in the notation following \eqref{eq:EBpm}, 
\begin{equation}
\label{eq;g1}
g_1 ( k, \underline \alpha) 
 : =  \int_{ \mathbb C/\Lambda } ( U (- z )  u_1(k,z)^2  - U ( z ) u_2(k,z)^2 ) dm ( z ) .
\end{equation}
An indirect argument presented in the proof of Proposition \ref{p:SB} shows that
if $ g_0 ( \underline \alpha ) \neq 0$ then $ g_1 (\underline \alpha ) \neq 0 $. 
This can also be verified numerically.

\subsection{Grushin problem for the self-adjoint Hamiltonian}
We now turn to the corresponding Grushin problem for $ H_k^B ( \alpha ) $ given in 
\eqref{eq:conj} (note the irrelevant change of sign of $ k $)
\begin{equation}
\label{eq:GrH} 
 \begin{gathered} 
\mathcal H_k^B ( \alpha , z  ) := 
\begin{pmatrix} H_k^B ( \alpha ) - z & \widetilde R_- (k )  \\
\widetilde R_+(k) & 0 \end{pmatrix} : 
H^1_0 ( \mathbb C/\Lambda ; \mathbb C^4 ) \times \mathbb C^2 \to 
L^2_0 ( \mathbb C/\Lambda ; \mathbb C^4 ) \times \mathbb C^2, \\ 
 H_k^B  ( \alpha ) := \begin{pmatrix} 0 & D_B ( \alpha ) ^* + \bar k  \\
D _B ( \alpha ) + k & 0 \end{pmatrix} , \\
\widetilde R_- ( k ) = \begin{pmatrix} 0 & R_+(k)^* \\
R_- ( k )  & 0 \end{pmatrix} , \ \ \
\widetilde R_+ ( k ) = \widetilde R_- ( k )^* , 
\end{gathered}
\end{equation}
where $ R_\pm ( k ) $ are the same as in \eqref{eq:Grushin1}. The operator
$ \mathcal H_k^B ( \alpha, z ) $ is invertible for all $ k$,  $ | B | \ll 1$, 
$ |\alpha - \underline \alpha | \ll 1$ and $ |z | \ll 1$. We denote the components
of the inverse by $ \widetilde E^B_\bullet ( k, \alpha , z ) $ and we have
\[   \widetilde E_+^0 ( k, \underline \alpha, 0 ) = 
\begin{pmatrix} 0 & E^0_+( k ) \\
E_-^0 ( k ) ^* & 0 \end{pmatrix} , \ \  \widetilde E_-^0 ( k, \underline \alpha, 0 ) = 
\widetilde E_+^0 ( k, \underline \alpha, 0 )^* , \ \ 
\widetilde E_{-+}^0 ( k, \underline \alpha, 0 ) \equiv 0  . \]
Using \cite[Proposition 2.12]{notes} again we see that (in the notation of
\eqref{eq:newE-+})
\[ \begin{split} \widetilde E^B_{-+}  ( k , \alpha, z ) & = 
\begin{pmatrix} z & E_{-+}^B ( k, \alpha) \\
E_{-+}^B ( k , \alpha )^* & z \end{pmatrix}  
+  \mathcal O ( |z|^2 + |B|^2 + |\alpha - \underline \alpha|^2 ) .
\end{split} \]
(Here we used the fact that
 $ E_-^0(k) E_-^0 ( k)^* \equiv 1 $ and 
 $ E_+^0 ( k )^* E_+^0 ( k )  \equiv 1 $ which follows from \eqref{eq:defE} and
 normalization of $ u(k) $ and $ u^* ( k ) $.)

Hence $ z = E_{1 }^B  ( k, \alpha ) = - E_{-1 }^B ( k ,  B) $ (the eigenvalues of $ H_k^B ( \alpha ) $ closest
$ 0 $) for $ k $ close to $ 0 $ are given by solutions of 
\begin{equation} 
\label{eq:det}
\begin{split}  & \det \widetilde E^B_{\pm } ( k , \alpha, z )  = 0 
\ \Longrightarrow \\
& \ \ \ \ \ \ \ \ \ \ \ \ \ \ \ \ \ \  z = \pm \left | \gamma_1 B k^2 + \gamma_0 ( \alpha - \underline \alpha ) + 
\mathcal O ( |B|^2 + | \alpha - \underline \alpha|^2  + |k|^4 ) \right | , \end{split} 
\end{equation}
where (under the assumption that $ g_0 ( \underline \alpha ) \neq 0 $) $ \gamma_0 \neq 0 $, 
$ \gamma_1 \neq 0 $. (The exact symmetry of signs follows from 
the extension of the chiral symmetry \eqref{eq:defW} to the Grushin problem \eqref{eq:GrH}
which 
shows that $ \det \widetilde E^B_{\pm } ( k , \alpha, z ) = \det \widetilde E^B_{\pm } ( k , \alpha, - z )$.)

\section{Bifurcation}
\label{s:bif}

\begin{figure}
{\begin{tikzpicture}
\node at (0,0) {\includegraphics[trim={2.1cm 0 0 0},width=7.5cm]{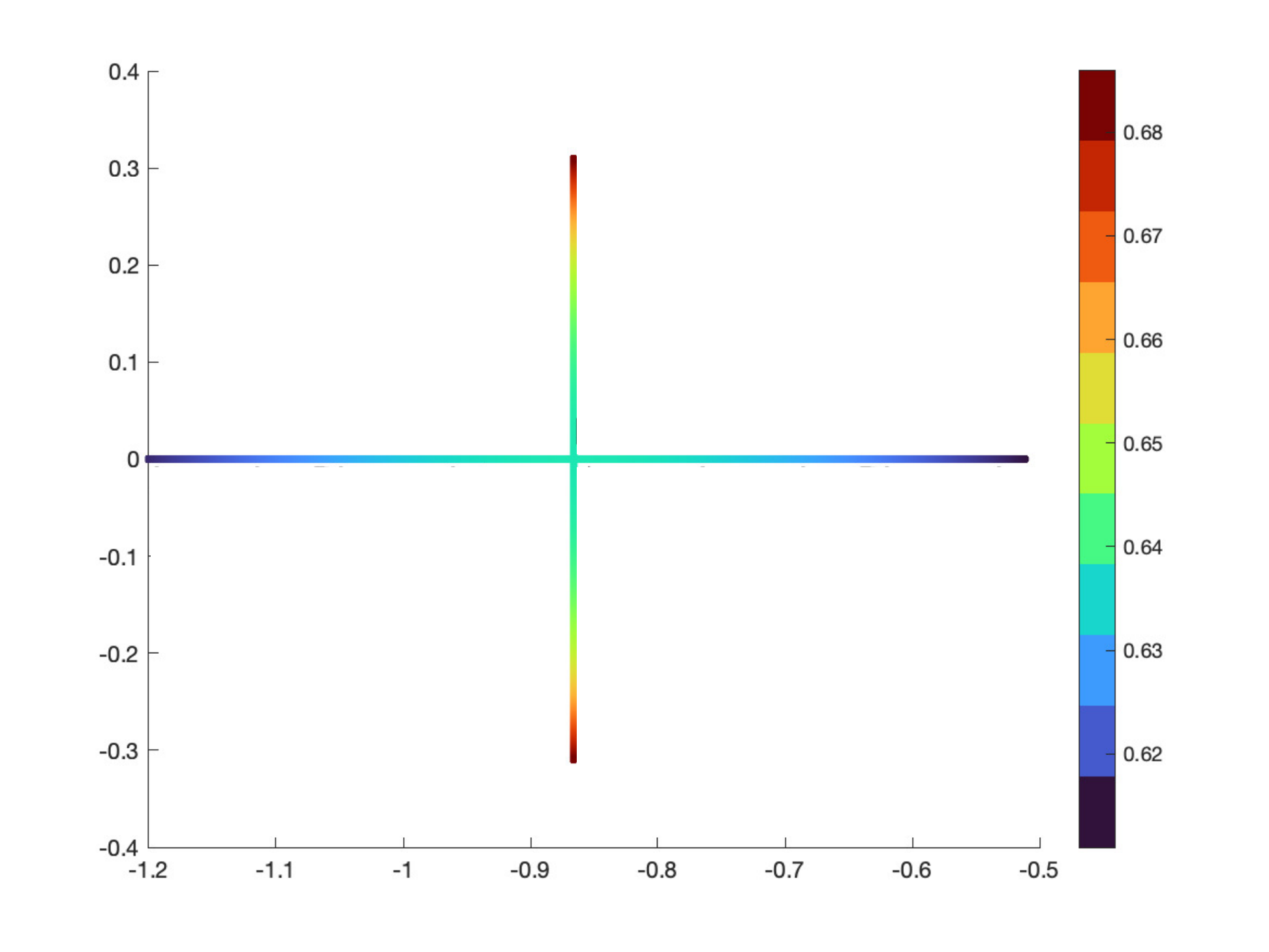}};
\node at (7.6,0){\includegraphics[trim={1.7cm 0 0 0},width=7.5cm]{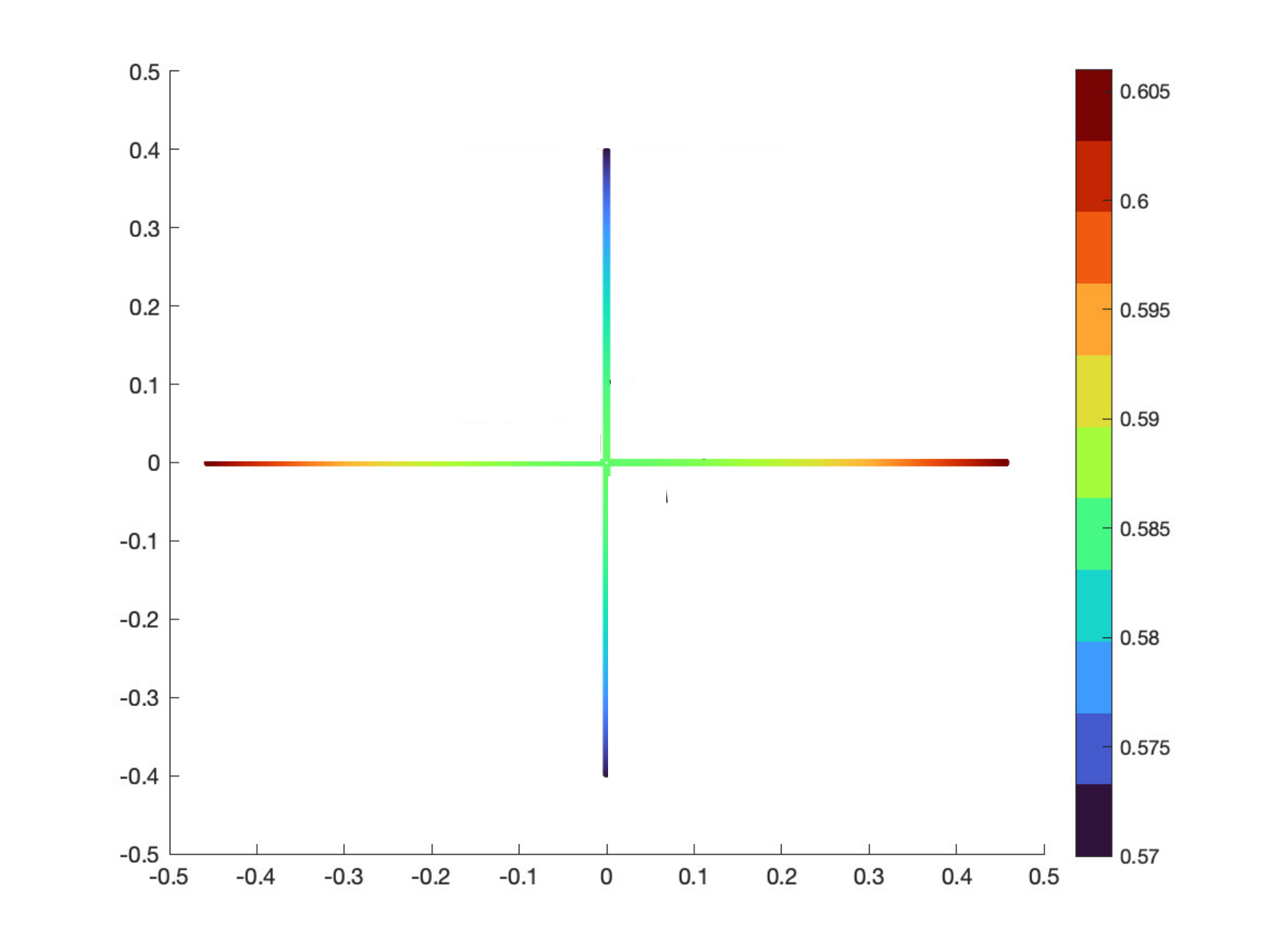}};
\draw [-{Stealth[length=3mm, width=2mm]}](-0.7,0)-- (-0.7,2.5);
\draw [-{Stealth[length=3mm, width=2mm]}](-0.7,0)-- (-0.7,-2.5);
\draw [-{Stealth[length=3mm, width=2mm]}](-2.9-0.6,-0.1)-- (-1,-0.1);
\draw [-{Stealth[length=3mm, width=2mm]}](2.7-0.6,-0.1)-- (-0.7,-0.1);
\node at (-0.7+1.3,0.3+0.2){\small Outgoing DPs};
\node at (-2-0.6,-0.4){\small Incoming DP};
\node at (2-0.6,-0.4){\small Incoming DP};
\draw [-{Stealth[length=3mm, width=2mm]}](-0.8+8,2.5)--(-0.8+8,0.1);
\draw [-{Stealth[length=3mm, width=2mm]}] (-0.8+8,-2.5) -- (-0.8+8,0.1);
\draw [-{Stealth[length=3mm, width=2mm]}](-1+8,-0.1)--(-2.9-0.6+8,-0.1);
\draw [-{Stealth[length=3mm, width=2mm]}] (-1+8,-0.1)--(2.7-0.6+8,-0.1);
\node at (-0.7+1.3+8,0.45){\small Incoming DPs};
\node at (-2-0.6+8,-0.4){\small Outgoing DP};
\node at (2-0.6+7.4,-0.4){\small Outgoing DP};
\node at (-1,-3) {$\Re k$};
\node at (-4.5,0.5) {$\Im k$};
\node at (7,-3) {$\Re k$};
\node at (2.5,-3) {$\alpha$};
\node at (7.6+2.5,-3) {$\alpha$};
   \end{tikzpicture}}
\caption{\label{fig:Bifurcation}Bifurcation for $B=0.1$. (Top): The color-coding indicates the position of the Dirac points for given values of $\alpha \in \mathbb R.$ The right figure illustrates the bifurcation at $\Gamma$ and the left figure at a non-equivalent (modulo $\Lambda^*$) bifurcation point that is a vertex of the boundary of the Brillouin zone, see Figure \ref{f:1}. }  \end{figure}
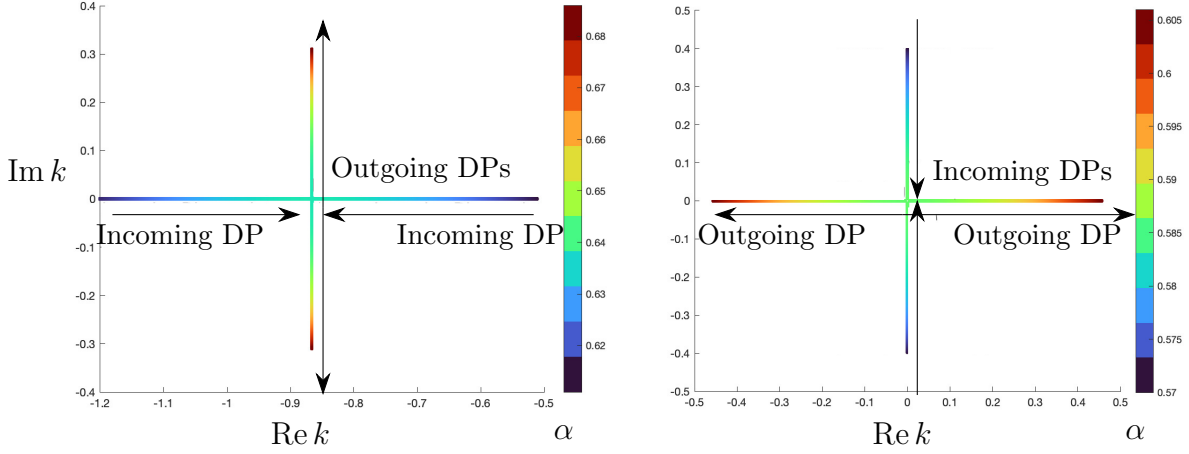
This section is devoted to showing \eqref{eq:bigcup} and giving a stronger version of Theorem \ref{th:lin}.
In view of \eqref{eq:specD}, for 
$k \notin \mathcal K_0 = \{K,-K\}+ \Lambda^*$, $ K = \frac 4 3 \pi > 0 $,  we have 
\begin{equation}
\label{eq:BS}
 \begin{gathered}  k \in \Spec_{L^2_0}(D_B(\alpha)) \Leftrightarrow 1/\alpha \in \Spec_{L^2_0}(T_k(B)), \\
T_k(B) := (2D_{\bar z}-k)^{-1} \begin{pmatrix} B & U(z) \\ U(-z) & -B \end{pmatrix}.
\end{gathered} 
\end{equation}
For $ k $ near $ \pm K $ (that is, near $ K $ and $K' $), 
and $ 0 < B \ll 1 $, we can use a modified
operator for which the equivalence in \eqref{eq:BS} still holds (near $ K $ and $ K'$):
\begin{equation}
\label{eq:wideT}  \widetilde T_k ( B ) := ( (2 D_{\bar z } - k)I_{\mathbb C^2 } + \mathcal B )^{-1}
\begin{pmatrix} 0 & U(z) \\ U(-z) & 0 \end{pmatrix}. \end{equation}
 In the notation of \eqref{eq:defTk}  $T_k(0)=T_k$  and its spectrum is given by 
 reciprocals of the magic $ \alpha$'s -- see Proposition \ref{p:magicD}. 

Combining the spectral characterization with the result of \S \ref{s:proof} we can 
obtain a rather precise characterization of the behaviour of eigenvalues of $ T_k ( B ) $:
\begin{prop}
\label{p:SB}
Suppose that $ \underline \lambda $ is a simple eigenvalue of $ T_k = T_k ( 0 )$
and that assumptions of Theorem \ref{th:Gamma} hold for $ \underline \alpha = 1/\underline \lambda$.
Then for
every $ \epsilon > 0 $ there exists $ \delta > 0  $ and a holomorphic function $ \lambda ( k, B ) $, 
such that $   \lambda ( k , B ) $ is a simple eigenvalue of $ T_k ( B ) $ and
\begin{equation}
\label{eq:propla}
\begin{gathered}
 ( k , B ) \mapsto \lambda ( k, B ) , \ \ \
 k \in \Omega_\epsilon :=  \mathbb C \setminus ( \mathcal K_0 + D ( 0 , \epsilon ) ), \ \ 
 B \in 
D ( 0 , \delta ) 
 \\
\lambda ( k + p  , B ) = \lambda ( k, B ) , \ \ \ p \in \Lambda^* , \ \ k , k+p \in 
\Omega_\epsilon  , \ \  B \in D ( 0 , \delta ) , \\
\lambda ( k, B ) = \overline  { \lambda ( \bar k, \bar B ) } = \lambda ( \omega k, \omega B ) = 
\lambda ( - k , B ) , \\
\lambda ( k, 0 ) = \underline \lambda, \ \ \ \partial_B \partial_k^2 \lambda ( 0 , 0 ) \in \mathbb R \setminus \{ 0 \}.
\end{gathered}
\end{equation}
In particular, for $B \in D ( 0 , \delta) \subset \mathbb C $, 
\begin{equation}
\label{eq:morela}  \lambda ( k , B ) = \underline \lambda + B^3 \lambda_0 (B^3 ) + c_1 B k^2 +  \mathcal O ( 
 B^4 k^2)  + \mathcal O ( B^2 k^4)  + \mathcal O ( B k^8 ) , \end{equation}
 where $  c_1 \in \mathbb R \setminus \{ 0 \} $ and $  \lambda_0 ( z )  =\overline{
\lambda_0 ( \bar z ) } $.  
\end{prop}

\noindent
{\bf Remarks.}  1. It follows from the proof that the constant $ c_1 $ can be computed using the constants 
$ g_0 ( \underline \alpha ) $ and $ g_1 ( \underline\alpha ) $ defined in 
\eqref{eq:G2g} and \eqref{eq:newE-+} respectively:
    \[ c_1 = -\frac{3\theta'(0)^2  }{16\pi^2 \theta (\tfrac12 )^2} \frac{g_0(\underline \alpha)}{g_1(\underline \alpha)}
.\]

\noindent
2. In view of \eqref{eq:BS}, \eqref{eq:morela}, shows that when $ \underline \alpha $ is
is magical and $ \underline \lambda = 1/\underline \alpha $ satisfies the assumptions of
Proposition \ref{p:SB}, then for $ 0 < |B| \ll 1 $, $  k \in D_B ( \alpha )$, if and only if
\begin{equation}
\label{eq:morela1} 
\begin{gathered} 
k^2 B  ( 1 +   f_0 ( k , B ) ) = 
c_1^{-1} \underline \alpha^{-2} ( \underline \alpha - \alpha -  \underline \alpha^2 B^3 \lambda_1  ( B^3 )) , \\
f_0 ( k , b ) =  \mathcal O ( 
 B^3)  + \mathcal O ( B k^2)  + \mathcal O ( k^6 ), \ \ \lambda_1 ( z ) = \lambda_0 ( z ) + \mathcal O ( z  ) .
 \end{gathered} 
 \end{equation}
 In particular, when $ B $ and $ \alpha $ are real then the eigenvalues of
 $ D_B ( \alpha ) $ bifurcate $ k = 0 $ when $ \alpha = \underline \alpha -
  \underline \alpha^2 B^3 \lambda_1  ( B^3) $ (we recall from \eqref{eq:morela} that 
  $ c_1 \in \mathbb R \setminus \{ 0 \} $ and $ \lambda _1 ( B^3 ) $ is real for $ B $ real).
  We see the same bifurcation for $ B = B_0 e^{ \pm 2 \pi i / 3} $, $ B_0 > 0 $,  obtained using \eqref{eq:symD}.

\noindent
3. Numerical evidence suggests (see Figure \ref{f:la0}) that $ \lambda_0 ( 0 ) < 0 $ for the 
Bistritzer--MacDonald potential. If $ B = B_0 e^{ 2 \pi i \theta } $ that means the $ \Gamma 
$ point (corresponding $ k = 0 $) is in the spectrum of $ D_B ( \alpha ) $, $ \alpha \in 
\mathbb R $, only if $ \theta \in \frac13 \mathbb Z $.

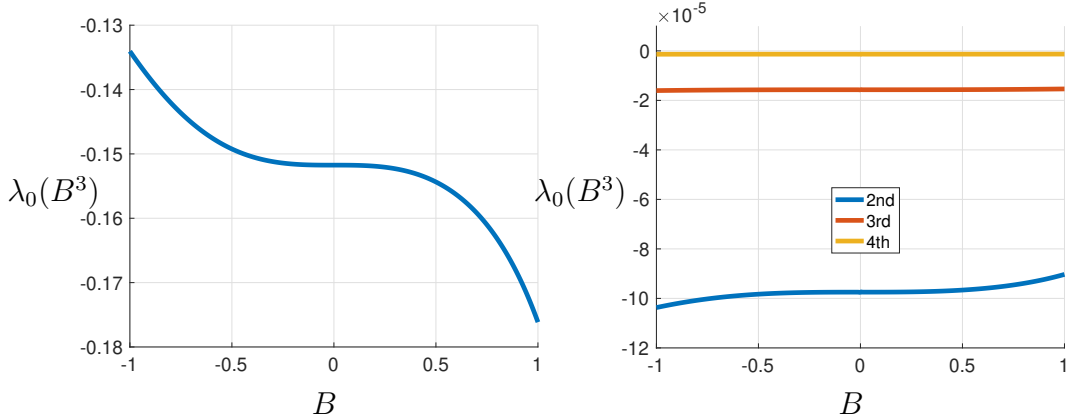
\begin{figure}
{\begin{tikzpicture}
\node at (0,0) {\includegraphics[width=7cm]{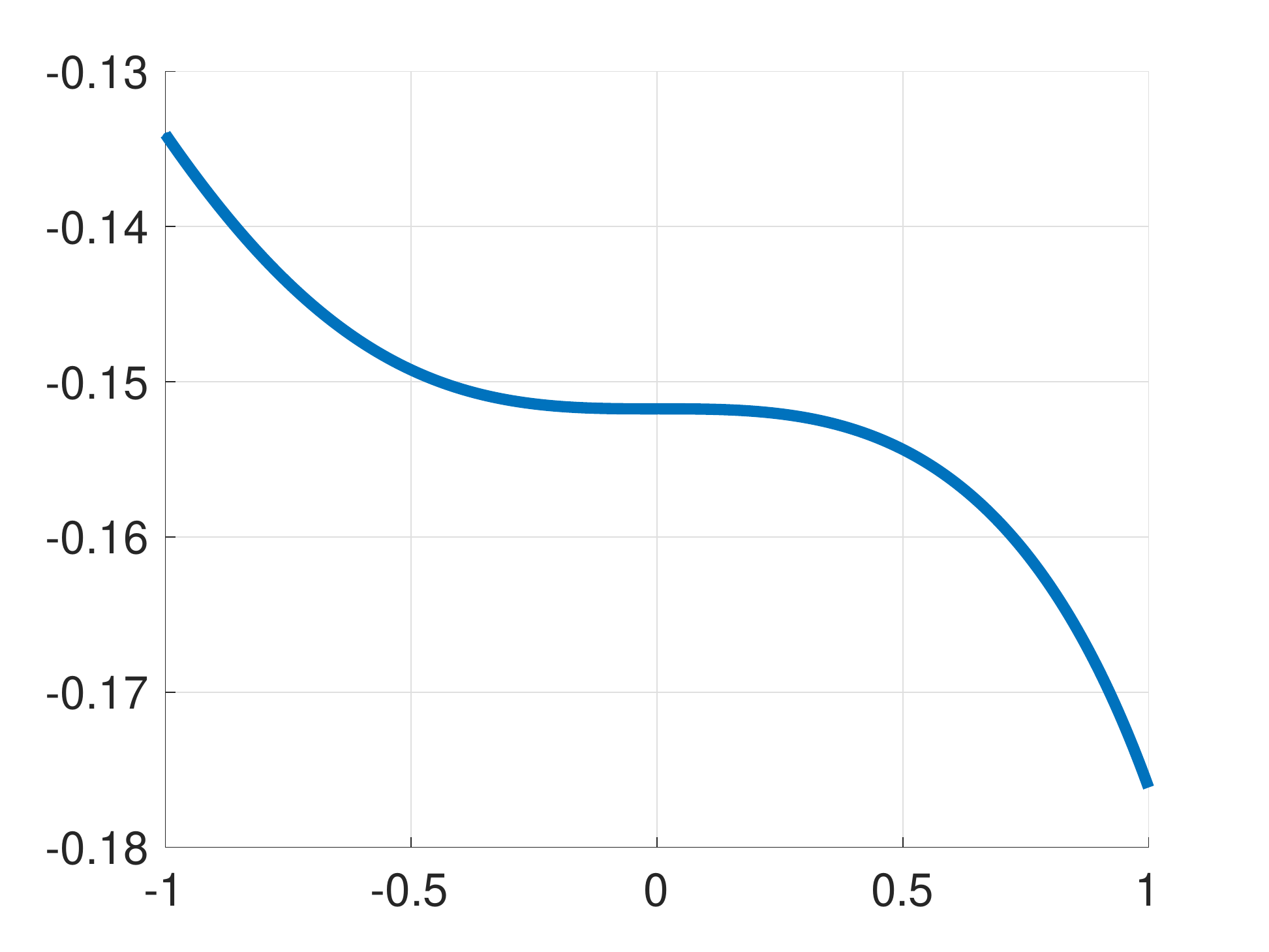}};
\node at (-3.6,0) {$\lambda_0(B^3)$};
\node at (0,-2.8) {$B$};
\node at (7,0) {\includegraphics[width=7cm]{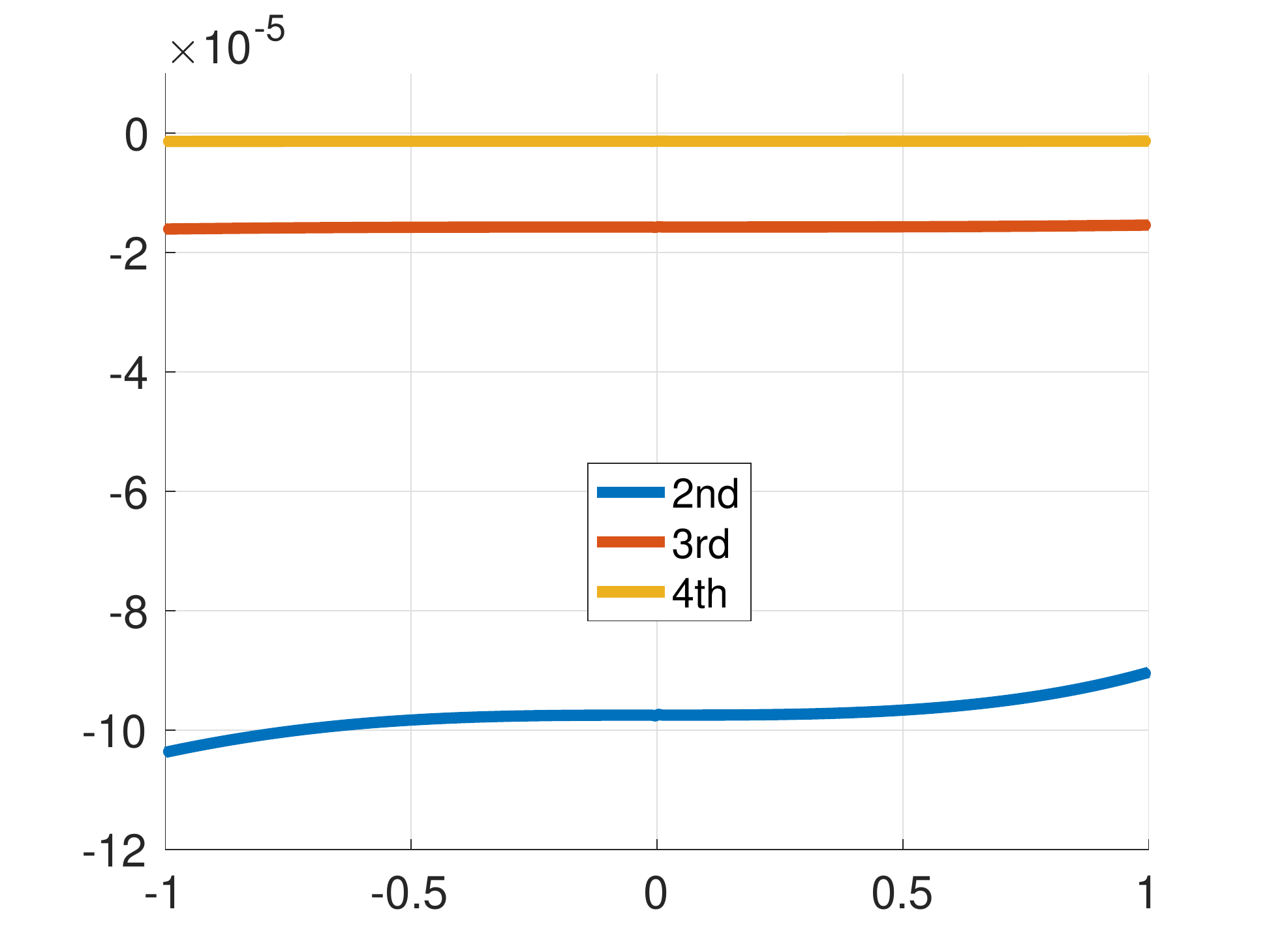}};
\node at (-3.6+7,0) {$\lambda_0(B^3)$};
\node at (7,-2.8) {$B$};
\end{tikzpicture}}
\caption{\label{f:la0}  Plot of $B \mapsto \lambda_0(B^3 )$ 
of Proposition \ref{p:SB} for the first (left) and second to fourth magic angle (right).}
\end{figure}

\begin{proof}[Proof of Proposition \ref{p:SB}]
Let  $ U \Subset \mathbb C \setminus \mathcal K_0 $ be an open set. Then for $ k \in U $, 
$ T_k ( B ) = T_k ( 0 ) + \mathcal O ( B )_{ L^2_0 \to L^2_0 } $ and if  $ 0 < \epsilon_0  \ll1 $,
then the projection, 
$   \Pi ( k , B ) := ( 2 \pi i )^{-1}
\int_{ \partial D ( \underline \lambda , \epsilon_0 ) } ( \zeta - T_k ( B ) )^{-1} d\zeta$, 
is holomorphic in $ k $ and $ B $ and has a fixed rank. We assumed
that $ T_k ( 0 ) $ has a simple eigenvalue at $ \underline \lambda $
(independent of $ k $ -- see \eqref{eq:specc}), which then implies that the rank is one, and
$ T_k ( B ) $ has a simple eigenvalue $ \lambda = \lambda ( k, B ) $. 
Since $ \lambda( k , B ) = \tr ( T_k ( B ) \Pi ( k , B ) ) $
it follows that $ \lambda ( k , B ) $ is holomorphic in $ k $ and $ B $. 

From  \eqref{eq:symD} and \eqref{eq:BS} we have 
\begin{equation} 
\label{eq:symT}  \Spec_{L^2_0}(T_k(B)) = \overline{\Spec_{L^2_0}(T_{\bar k}(\bar B))} = \Spec_{L^2_0}(T_{-k}(B)) = \Spec_{L^2_0} ( T_{\omega k } ( \omega B ) ) ,\end{equation}
which gives 
\begin{equation}
\label{eq:symla}    \lambda ( k , B ) = \lambda ( -k , B ) = \overline{\lambda ( \bar k , \bar B ) } =
\lambda ( \omega  k , \omega B ) . 
\end{equation}
From \eqref{eq:BS} and the periodicity of the spectrum of $ D_B ( \alpha ) $ with 
respect to $ \Lambda^* $ (see Lemma \ref{l:ev2ker}) we also note that
\begin{equation}
\label{eq:perla}  \lambda ( k + p , B ) = \lambda ( k , B ) ,  \ \ p \in \Lambda^*  
\end{equation} 
provided that $ k , k + p \in U $. This allows an extension of $ \lambda (k, B )  $ to $ \Omega_\epsilon $
in the statement of the proposition, provided that $ |B | < \delta $ for some sufficiently small $ \delta $. 
The properties of the expansion \eqref{eq:morela} come from the fact that individual terms
in the Taylor expansion satisfy the symmetries \eqref{eq:symla}:
\begin{equation}
\label{eq:symmetries}
\begin{split} 
 & a k^p B^q = a (-1)^p k^p B^q = \overline a k^p B^q = a \omega^{p+q} k^p B^q 
\ \Longrightarrow \\
& \ \ \ \ \ \ \ \ \ \ \ \ \ \ \ \ \ \ \ \ \  \ \ \ \ \ \ \ \ \ \ \ \ \ \ \ \ \ \ a \in \mathbb R, \  p \in 2 \mathbb N , \  -p \equiv q \hspace{-0.35cm} \mod 3 . 
\end{split} 
\end{equation}
This proves \eqref{eq:propla} and \eqref{eq:morela} except for the non-degeneracy of $ \partial_B \partial^2_k 
\lambda ( 0 , 0 ) $. 

To prove it, we compare the Taylor expansion of $ \lambda ( k , B ) $ with the effective
Hamiltonian \eqref{eq:newE-+}. Hence, with  $ \mu :=
1/\alpha $, $ k \in \Spec_{L^2_0} D_B ( \alpha ) $ then (with some modification of notation)
\begin{equation}
\label{eq:Bg0} 
B g_0 ( \underline \alpha ) \theta_2 ( 0 )^{-2} ( \theta ( z ( k ) )^2 + B f_1 (k, B ) ) + 
\sum_{ p=1}^P
( \mu  - \underline \lambda )^p B^{k_p}  F_p ( k , B , \mu ) = 0, 
\end{equation}
where $ F_P ( 0, 0 , \underline \lambda ) \neq 0 $ and $ k_P = 0 $, 
$ k_p > 0 $, $ p < P $.  That has to be so by noting that $ 0 \notin \Spec_{L^2 } ( 
D_B ( \alpha )) $ for $ 0 < | \alpha - \underline \alpha | \ll 1 $. 
Hence, (note that $ \lambda ( k , B ) $ is even in $ k $)
\[ \begin{split} ( \lambda ( k , B ) - \underline \lambda )^P &  =  c_0 B k^2 +  B^2 \tilde f_1 ( k,  B) + 
B k^4 \tilde f_2 ( k , B ) \\
& \ \  + \sum_{p=1}^{P-1} ( \lambda ( k ,  B ) - \underline \lambda )^p B^{k_p} \widetilde F_p ( k, B , \lambda ( k ,  B ) )
 , \ \ c_0 \neq 0. \end{split} \]
Combined with \eqref{eq:morela} this gives 
\[ ( B^3 \lambda_0 (  B^3 ) + c_1 B k^2 +  \mathcal O ( B^2 k^2 + 
 B^2 k^4 + B k^8 ) )^P = c_0 B k^2 + \mathcal O ( B^2 ) + \mathcal O ( B k^4) , \ \ , \] 
 which implies that $ P = 1 $ and $ c_1 =c_0 $, and concludes the proof.
 
We also note that comparing this conclusion with \eqref{eq:newE-+} 
shows that $ g_1 ( \underline \alpha ) \neq 0 $. 
\end{proof} 

At the bifurcation point, the Bloch eigenvalues exhibit a quadratic well, see Figure \ref{f:c2w}.
\begin{prop}
\label{prop:well}
Under the assumptions, and in the notation, of Proposition \ref{p:SB} and \eqref{eq:morela1}, 
let 
\[ \alpha^* = \underline \alpha  +  \underline \alpha^2 B^3 \lambda_1 ( B^3 ) , \]
so that $ 0 \in \Spec_{L^2_0 } D_B ( \alpha^* ) $. 
Then the two Bloch eigenvalues $E_{\pm 1}$ of $H_k^B(\alpha)$ closest to zero, defined in \eqref{eq:eigs}, satisfy
\[ E_{\pm 1}(\alpha^*,k) = \pm \vert \gamma_1 B k^2\vert + \mathcal O(B^2 +\vert k \vert^4 ), \ \ \gamma_1 >  0 . \]
\end{prop}
\begin{proof}
This follows from \eqref{eq:det} and \eqref{eq:morela1}.
\end{proof}

The next proposition deals with the vertices of the boundary of the  Brillouin zone.
In view of \eqref{eq:perla} it is enough to consider one of the vertices, say,
\begin{equation} 
\label{eq:defk1} \underline k_1 :=   2 \pi i /\sqrt 3 . \end{equation}
We then have 

\begin{prop}
\label{p:SB1}
Suppose that $ \underline \lambda $ is a simple eigenvalue of $ T_k = T_k ( 0 )$
and that assumptions of Theorem \ref{th:Gamma} hold for $ \underline \alpha = 1/\underline \lambda$.
Then, for $ k $ near $ \underline k_1  $ given in \eqref{eq:defk1}, 
\begin{equation}
\label{eq:morela2}  \lambda ( k , B ) =\underline \lambda + B \lambda_2 (  B ) + 
c_2 B^{q} (k - \underline k_1 )^2   + \mathcal O ( 
 |B|^{q+1} |k- \underline k_1 |^2 )
  \end{equation}
 where $q \geq 1$, $ c_2   \in \mathbb R \setminus \{ 0 \} $
 and $  \lambda_2 ( z )  =\overline{
\lambda_2( \bar z ) }.$ 
\end{prop}

\noindent
{\bf Remarks.} 1. We again have a bifurcation result similar to \eqref{eq:morela1} but less precise:
\begin{equation}
\label{eq:morela3} 
\begin{gathered} B^q (k-\underline k_1)^2 ( 1 + f_1 ( k , B ) ) = c_3^{-1} \underline \alpha^{-2} ( 
\underline \alpha - \alpha - B \lambda_3 ( B ) ) , \\
c_3 \in \mathbb R \setminus \{ 0 \} ,  \ \ 
\ {\overline{\lambda_3 ( \bar z  ) } = \lambda_3 ( z )} , \ \ 
f_1 ( k , B ) = \mathcal O ( B ) . 
\end{gathered}
\end{equation}
For $ B $ real we see a bifurcation at $ \alpha_* = \underline \alpha + B \lambda_3 ( B ) $,
with similar bifurcations for $ B = B_0 e^{ \pm 2 \pi i / 3} $, $ B_0 > 0 $,  obtained using \eqref{eq:symD}.

\noindent
2. If we know (which can be checked numerically for the Bistritzer--MacDonald
potential) that $ g_1 ( \underline k_1 , \underline \alpha ) \neq 0 $ and $ q =1 $ then
\[ c_2 = \frac{ g_0 ( \underline \alpha ) } { g_1 ( \underline k_1 , \underline \alpha ) }  . \]
This follows from a comparison with the effective Hamiltonian in \eqref{eq:newE-+}.

\begin{proof}[Proof of Proposition \ref{p:SB1}]
From \eqref{eq:symla} and periodicity, $\lambda(k+\Lambda^*,B) = \lambda(k,B)$, 
we conclude that {(note that $ 2 k_1 =  4 \pi i /\sqrt 3 \in \Lambda^* $)
\begin{equation}
\begin{split}
\lambda(\underline k_1 + z ,B) &= \lambda(-\underline k_1-z ,B) =\lambda(\underline k_1 - z ,B) \\
&= \overline{ \lambda(- \bar z - \underline k_1,\bar B)} =\overline{ \lambda( \underline k_1 - \bar z ,\bar B)}.
\end{split}
\end{equation}
We also note that for $ k\notin  \mathcal K_0 + D ( 0 , \epsilon ) $, $ \lambda ( k , 0 ) = 
\underline \lambda $ (since $ k \in \Spec_{L^2_0 } D_0 ( \alpha ) $ only at $ \alpha = 
\underline \alpha = 1/\underline \lambda $. 
Hence, 
 \begin{equation}
 \label{eq:lamb2} 
 \begin{gathered} \lambda(k,B) = \underline \lambda 
 + B \lambda_2 ( B ) + c_2 B^{q} ( k- \underline k_1)^2 
 + \mathcal O(B^{q+1} (k- \underline k_1)^2 ) , \\
 c_2 \in \mathbb R , \ \ \  \lambda_2 ( \bar z) = \overline {\lambda_2 ( z ) } 
\end{gathered}
\end{equation}
Suppose that $ c_2 = 0 $. 
\{Proposition \ref{p:SB} already shows that $ \lambda ( k , B ) $ cannot
be independent of $ k $. This and and  $ c_3 = 0 $  imply that 
for some $  q \geq 1$, $ p > 1 $, 
\[ \begin{split} \lambda(k,B) & = \underline \lambda + B c_2+  B^2 \lambda_2 ( B ) + c_4 B^q (k-\underline k_1 )^{2p} ( 1 +  \mathcal O ( B) )   . \end{split} \]}
Now fix $ \mu = \lambda ( k , B ) $ close $ \underline \lambda $. Then
for a fixed $ 0 < |B| \ll 1 $ 
 we would have $ 2 p$ solutions $ k $ near $ \underline k_1 $ (recall that $ k \mapsto \lambda ( k, B ) $
 is holomorphic near $ \underline k_1 $). 
 However, Theorem \ref{th:Gamma} shows that there are at most two solutions for a fixed
 $ \lambda (k  , B ) $. {This gives a desired contradiction and shows that $ c_2 \neq 0 $.}
\end{proof}

\begin{figure}
{\begin{tikzpicture}
\node at (0,0) {\includegraphics[width=13.5cm,height=8cm]{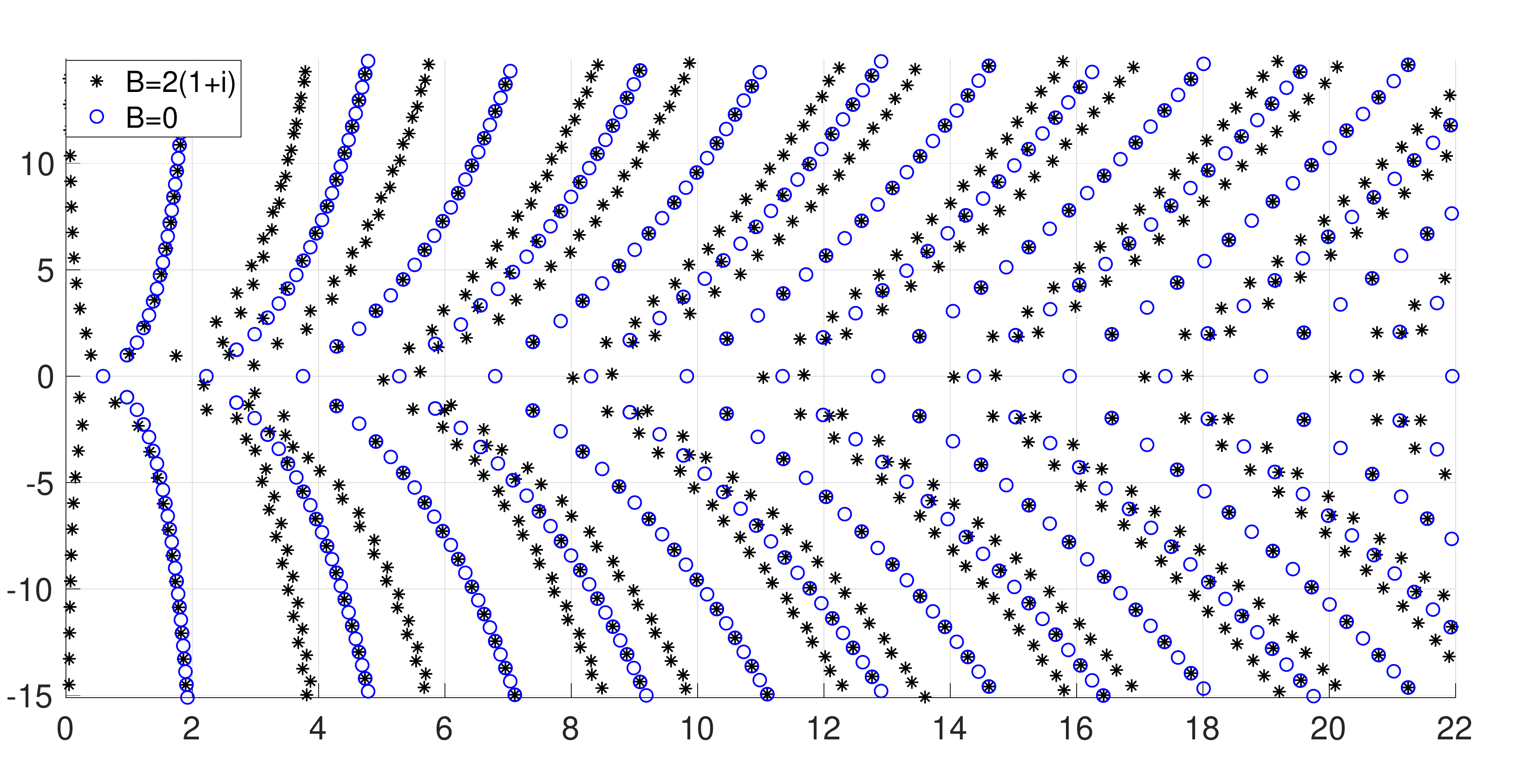}};
\node at (0,-4) {$\Re \alpha$};
\node at (-7,0) {$\Im \alpha$};
\end{tikzpicture}}
{\begin{tikzpicture}
\node at (0,-4) {\includegraphics[width=13.5cm,height=8cm]{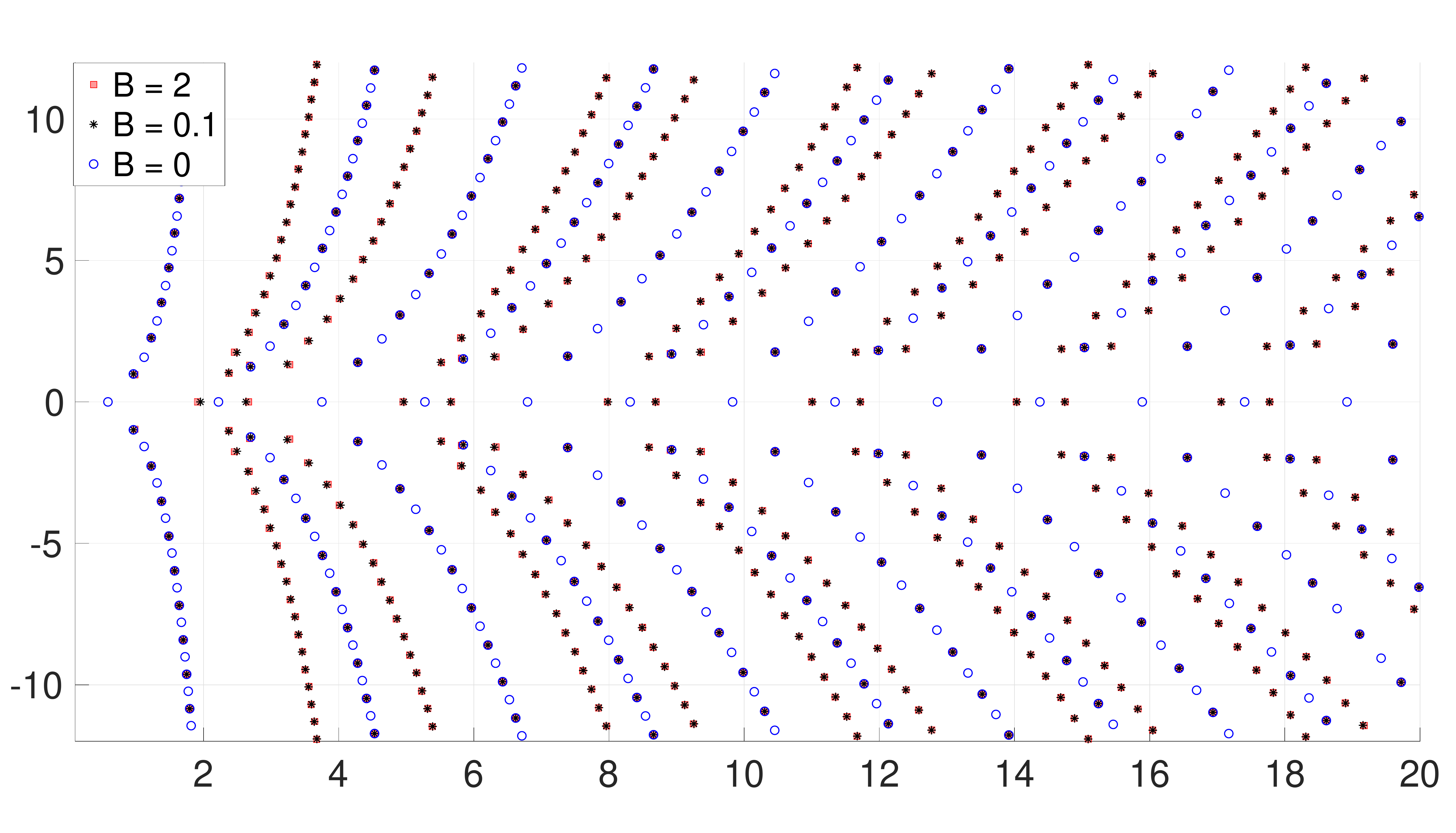}};
\node at (0,-8) {$\Re \alpha$};
\node at (-7,-4) {$\Im \alpha$};
\end{tikzpicture}}
\caption{\label{fig:Dirac} Top figure showing $\alpha \in \CC$ such that $1/\alpha \in \Spec_{L^2_0}(\widetilde T_K(B))$ or $K \in \Spec_{L^2_0}(D_B(\alpha)).$ We see that indeed for $B \in \mathbb R \setminus \{0\}$ the trajectory of Dirac points passes through $K,K'.$
Bottom figure showing $\alpha \in \CC$ such that $1/\alpha \in \Spec_{L^2_0}(\widetilde T_K(B))$ or $K \in \Spec_{L^2_0}(D_B(\alpha)).$ For general $B \notin \RR$ the trajectory of Dirac points for varying $\alpha \in \mathbb R$ does not pass through $K$ between successive real magic angles.}
\end{figure}
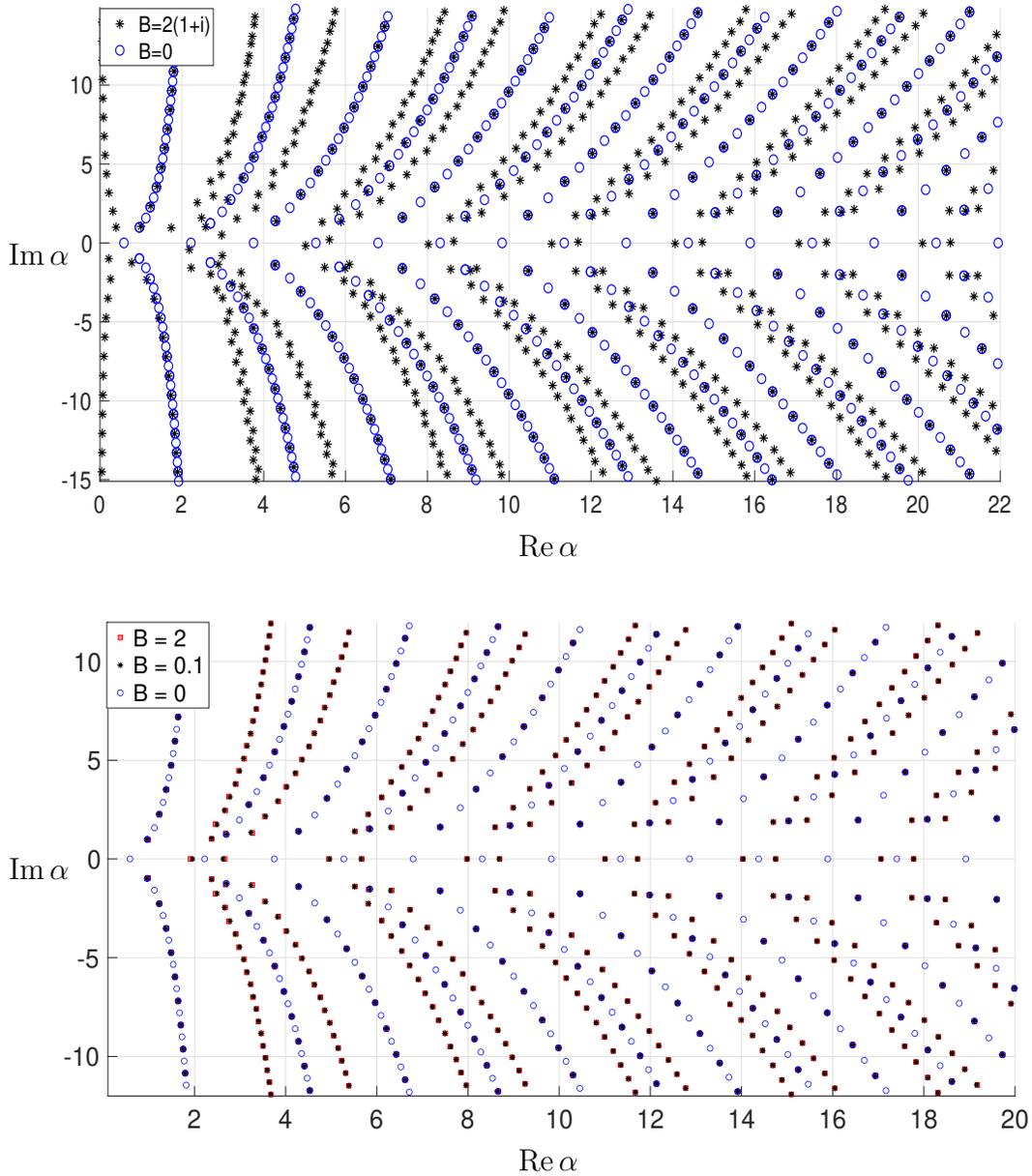



\section{Proofs of Theorems \ref{th:Gamma} and \ref{th:lin}}
\label{s:proofs}

Combining the results of previous of sections we can now prove the main 
results of this paper.

\begin{proof}[Proof of Theorem \ref{th:Gamma}]
In the notation of \eqref{eq:newE-+}  we see the effective Hamiltonian for 
$ D_B ( \underline \alpha ) $ for $ B $ small 
\begin{equation} 
\label{eq:newnewE} E_{-+}^B (k , \alpha ) = - B c( k ) c^* ( k )  ( c_0 \theta ( z ( k ) )^2 + \mathcal O ( B ) ) 
+ \mathcal O ( \alpha - \underline \alpha)  
. \end{equation}
Since $ \theta ( z ( k ) ) \neq 0 $ for 
$ k \notin \Lambda^* $ (see \S \ref{s:theta}) 
there exists a constant $ a_1 $ such that if $ |\alpha - \underline \alpha| < a_1 |B| $
then $ E_{-+}^B ( k ) $ is not identically $ 0 $ (provided that $ B $ is small enough). 
This shows invertibility at some $ k $ and hence discreteness of the spectrum
(by the analytic Fredholm theory applied to $ k \mapsto  ( D_B ( \alpha ) - k )^{-1} $
-- see for instance \cite[Theorem C.8]{res}) for 
\begin{equation}
\label{eq:Om1}   ( B , \alpha ) \in \Omega_1  := \{ ( B , \alpha ) : |B| < \delta_1, \ \ 
| \alpha - \underline \alpha | < a_1 |B |\}.  
\end{equation}
On the other hand, we can put $ k = 0 $ and recall from the proof of
Proposition \ref{p:SB} (see \eqref{eq:newE-+}) that
\[ E_{-+}^B ( 0, \alpha ) = c_0 ( \alpha - \underline \alpha ) ( 1 + 
\mathcal O ( \alpha - \underline \alpha ) + \mathcal O ( B ) ) + \mathcal O ( B^2 )  , \ \ c_0 \neq 0 .\]
Hence $ E_{-+}^B ( 0 , \alpha ) $ does not vanish if, for some constant $ A_1 $, and small $\delta_2>0$, 
\begin{equation}  ( B , \alpha ) \in \Omega_2  := \{ ( B , \alpha ) : \ \ 
 A_1 |B|^2 < | \alpha - \underline \alpha | < \delta_2 \}. 
\end{equation}
Again, that implies discreteness of the spectrum. We now note that there exists $ \delta_0 > 0 $ such that
\[  (D( 0 , \delta_0 ) \setminus \{ 0 \} )  \times D ( \underline \alpha , \delta_0 ) \subset \Omega_1 \cup \Omega_2 , \]
and this proves discreteness of the spectrum of $ D_B ( \alpha ) $ for $ 0 < |B |  < \delta_0 $
and $ | \alpha - \underline \alpha | < \delta_0 $. 

We also see that \eqref{eq:newnewE} implies \eqref{eq:nearG}: for $ U \Subset \mathbb C $, 
for any epsilon there exits $ \rho > 0 $ such that 
$ |\theta ( z ( k ) )^2| > \rho $ for $ z \in \mathcal U \setminus ( \Lambda^* + D( 0 , \epsilon ) ) $. 
But then 
\[ | E_{-+}^B ( k , \alpha ) | >    c_0 c( k ) c^* ( k ) |B | \rho - \mathcal O ( B^2 ) - \mathcal O
( |\alpha - \underline \alpha | ) > 0 , \]
if $  0 < |B| \leq \rho/C  $ and $ | \alpha - \underline \alpha | <  \rho |B | /C $ for some
(large) constant $ C $. 

It remains to prove \eqref{eq:count}. Let $ F $ be a fundamental domain of $ \Lambda^* $
containing $ 0 $ and such that there are no eigenvalues on $ \partial F $ (that can 
be arranged as under our assumptions the spectrum of $ D_B ( \alpha ) $ 
 is discrete and periodic with respect to $ \Lambda^* $).
Then, 
\[  \vert \Spec_{L^2_0}(D_B(\alpha)) \cap \CC/\Gamma^* \vert =
\frac{1}{ 2 \pi i } \tr \int_{ \partial F }  ( \zeta - D_B ( \alpha ) )^{-1} d \zeta . \]
As long $ D_B ( \alpha ) $ has no eigenvalue on $ \partial F $ for $ ( B , \alpha) \in
K \subset \mathbb C^2 $, this value remains constant for $ ( B, \alpha ) \in K $.
Choosing a small $ \epsilon $ and $ \delta $ needed for \eqref{eq:nearG} and 
putting  $ K = \{ ( B, \alpha ) : |B | < \delta,\  | \alpha - \underline \alpha |< a_0 \delta | B | \} $, 
we see that (using \cite[Proposition 4.2]{SZ})
\[ \begin{split} \frac{1}{ 2 \pi i } \tr \int_{ \partial F }  ( \zeta - D_B ( \alpha ) )^{-1} d \zeta & =
\frac{1}{ 2 \pi i } \tr \int_{ \partial D( 0 , \epsilon) }  ( \zeta - D_B ( \underline \alpha ) )^{-1} d \zeta
\\
& = \frac{1}{ 2 \pi i }  \int_{ \partial D( 0 , \epsilon ) } E_{-+}^B ( \zeta  )^{-1} d_\zeta E_{-+}^B 
( \zeta )  \\
& =  \frac{1}{ 2 \pi i }\int_{ \partial D( 0 , \epsilon ) } ( \zeta^2 + \mathcal O ( B ) )^{-1} ( 2 \zeta + \mathcal O 
( B ) ) d \zeta \\
& = 2 + \mathcal O  (B ) = 2 , \end{split}  \]
provided $ B $ is small enough (depending on $ \epsilon $, note that $ \alpha = \underline
\alpha $ in the calculation; the answer has to be an integer). 

We now need to account for the possibility that $ D_B ( \alpha ) $ has an 
eigenvalue on $ \partial F $. Periodicity of the spectrum shows that if $ k_1 \in 
\Spec D_B ( \alpha ) \cap \partial F $ then $ k_1 + \gamma \in \partial F $ 
for a finite number of  $ \gamma \in \Lambda^* $ (from the definition of a fundamental domain).
Only one of these points can be in the fundamental domain $ F $ and a small deformation
includes it in the interior of (the new) $ F$, while excludes all others from $ \partial F $.
The previous argument shows that the number of eigenvalues remains $ 2$.
\end{proof}

\begin{proof}[Proof of Theorem \ref{th:lin}]
When $ B, \alpha \in \mathbb R  $ then 
the last identity in \eqref{eq:symD} gives
\begin{equation}
\label{eq:real}  \Spec_{ L^2_0 } D_{   B}  ( \alpha )  = - \Spec_{L^2_0 } D_B ( \alpha ) = 
\overline{ \Spec_{L^2_0 } D_B ( \alpha ) }. \end{equation}

From Theorem \ref{t:symD} we know that for $ \alpha \notin \mathcal A $, 
\[ \Spec_{ L^2_0 } ( D_B ( \alpha ) ) = \{ d ( \alpha ), - d( \alpha )  \} + \Lambda^* \]
 (we fix $ B \in 
\mathbb R $ here) and 
\eqref{eq:real} shows that
\[    \overline {d ( \alpha )} \equiv  d ( \alpha ) \! \! \! \mod \Lambda^* \ \text{ or } \ 
 \overline {d ( \alpha )} \equiv  - d ( \alpha ) \! \! \! \mod \Lambda^* . \]
 Since $ \overline{ \Lambda^* } = \Lambda^* $ this means that 
$  \Spec_{L^2_0 } D_B ( \alpha ) \subset ( \mathbb R + \Lambda^* ) \cup ( i \mathbb R + 
\Lambda^* )  $ which is the same as \eqref{eq:rect1}.  

To prove \eqref{eq:bigcup} we recall that $\CC \times (\CC\setminus \mathcal K_0) \ni (B, k) \mapsto T_k(B)$, is a holomorphic family of compact operators with simple eigenvalue $\mu=1/\underline \alpha \in \Spec(T_k(0))$. We define $\mathscr K:=\mathscr R \setminus \bigcup_{k' \in \mathcal K_0}  B (k', \epsilon )$, then by periodicity of the spectrum of $D_B(\alpha)$ it suffices to restrict us to a fundamental domain: Since $\mathscr K/\Lambda^*$ is a compact set, the spectrum of $\mathscr K \ni k \mapsto T_k(B)$ is uniformly continuous on $ B $ in compact sets. Thus, for $0<\vert B \vert <\delta_0$ small enough, the operator $T_k(B)$ has precisely one eigenvalue in a $\delta_1$ neighbourhood of $\mu$ for every $k.$ This implies that for every $k \in \mathscr K/\Lambda^*$ there is precisely one $\mu_k$ such that $\mu_k \in \Spec(T_k(B))$ and $\vert \mu_k - \mu \vert < \delta_1.$ From Propositions \ref{p:SB} and \ref{p:SB1} we conclude that $\mu_k \in \mathbb R$ and the result follows.
\end{proof}

\noindent
{\bf Remark.} While our proof does not show that for $B \in \mathbb R\setminus \{0\}$ the points $K,K'$ are also in the spectrum of $D_B(\alpha)$ for some real $\alpha$ between successive magic angle, the bottom figure in Figure \ref{fig:Dirac} shows that this is indeed the case. For general $B \notin \RR$ this is however false, as the top figure in Figure \ref{fig:Dirac} shows. Both figures exhibit an interesting universal pattern for $\vert \alpha \vert $ large.

\end{document}